\documentclass[preprint,12pt]{elsarticle}

\usepackage{amssymb}
\usepackage{multirow}
\usepackage{lineno}
\usepackage[colorlinks,citecolor=black,linkcolor=black,urlcolor=black,bookmarks,bookmarksnumbered]{hyperref}
\usepackage{amsthm}
\usepackage{amsmath}

\newtheorem{theorem}{Theorem}[section]
\newtheorem{remark}[theorem]{Remark}
\newtheorem{lemma}{Lemma}[theorem]

\newtheorem{definition}{Definition}[section]

\journal{Journal of Process Control}

\begin{document}
	
	%
	\begin{frontmatter}
		
		
		
		\title{Self-tunable approximated explicit MPC: Heat exchanger implementation and analysis}
		
		
		\author{Lenka Galčíková}
		\author{Juraj Oravec}
		
		\affiliation{organization={Slovak University of Technology in Bratislava, Faculty of Chemical and Food Technology, Institute of Information Engineering, Automation, and Mathematics,},
			            addressline={Radlinského 9}, 
			            city={Bratislava},
			            postcode={81237}, 
			            country={Slovakia}
			}
		
		\begin{abstract}						
			The tunable approximated explicit model predictive control (MPC) comes with the benefits of real-time tunability without the necessity of solving the optimization problem online. This paper provides a novel self-tunable control policy that does not require any interventions of the control engineer during operation in order to retune the controller subject to the changed working conditions. Based on the current operating conditions, the autonomous tuning parameter scales the control input using linear interpolation between the boundary optimal control actions. The adjustment of the tuning parameter depends on the current reference value, which makes this strategy suitable for reference tracking problems. Furthermore, a novel technique for scaling the tuning parameter is proposed. This extension provides to exploit different ranges of the tuning parameter assigned to specified operating conditions. The self-tunable explicit MPC was implemented on a laboratory heat exchanger with nonlinear and asymmetric behavior. The asymmetric behavior of the plant was compensated by tuning the controller's aggressiveness, as the negative or positive sign of reference change was considered in the tuning procedure. The designed self-tunable controller improved control performance by decreasing sum-of-squared control error, maximal overshoots/undershoots, and settling time compared to the conventional control strategy based on a single (non-tunable) controller. 
		\end{abstract}
		%
		%
		
		
		\end{frontmatter}
	
	
	\section{Introduction}
	\label{sec:introduction}

	The current crisis of energy resources emphasizes the long-term goal of achieving sustainable industrial production and optimal energy utilization. Moreover, minimizing the energy utilization directly reduces the corresponding CO$_{2}$ emissions. Therefore, sustainable industrial production is focused on the wide implementation of advanced control methods~\cite{MN20}. A recent survey on applied thermal engineering focused on energy saving and pollution reduction from the industrial perspective is provided in~\cite{YV16}, and references therein.
	
	The heat exchangers in their numerous variants are integrated into many industrial plants as the heat transfer represents the crucial phenomena for all thermal energy applications~\cite{KV18}. Simultaneously, the utility generation for heating or cooling is energy-demanding. From the control viewpoint, the controller design for the heat exchangers is a challenging task due to the necessity to take into account the nonlinear and asymmetric behavior of the device, i.e., different plant behavior when the temperature is increasing, in contrast to the behavior when the temperature is decreasing, see~\cite{RL20}. 
	
	
	A very common challenge in terms of the time-varying behavior of heat exchangers is fouling. The authors in~\cite{AGUEL_fouling} focus on modeling the thermal efficiency in a cross-flow heat exchanger using an artificial neural network, which leads to a highly accurate model. In~\cite{TRAFCZYNSKI_fouling}, the authors address the effect of fouling by adjusting the parameters of the proportional–integral–derivative (PID) controller. 
	
	Although the conventional and widely-used PID controllers are robust and easy to tune, their control performance may not be sufficient. Various extensions built above the well-tuned PID controller were developed to compensate for the nonlinear and asymmetric behavior, often affected by the additional impact of the uncertain parameters. Such widely-used control strategies include, e.g., the robust control~\cite{WY18}, the gain-scheduling, and adaptive control. In a recent study~\cite{VANNIEKERK_tuning}, the authors suggest to adjust the controller online, based on a minimization of an objective function designed to achieve the desired control performance. 
	For the rigorous mathematical modeling and controller design methods in general, see~\cite{MF08}, and for the controller design tailored for the process control engineers see~\cite{Liptak}.
	
	One of the promising control strategies addressing all these issues in an optimal way came with the formation of the model predictive control (MPC), e.g., see~\cite{Morari_MPC}. MPC provides optimal control input based on the minimization of a specified cost function while considering a model of the system. Compared to linear quadratic controllers (LQR)~\cite{LQR}, model predictive control also includes constraints on the control input or process variables~\cite{Maciejowski_MPC}, and additional saturation is not necessary. Moreover, as the optimization problem is solved in every control step, MPC represents a receding horizon control policy~\cite{receding_horizon}, having a significant benefit mainly in the terms of disturbance rejection. The model predictive control was intensively investigated in connection with heat exchangers. In~\cite{Vinaya_HE_MPC}, the authors developed a model predictive control for a shell and tube heat exchanger. 
	Four robust control strategies were presented and compared in~\cite{Oravec_HE_ATE}. A two-level control structure was applied on a heat exchanger network in~\cite{Gonzales_HE_MPC}, where the low level of control was ensured by MPC and the high level by a supervisory online optimizer. The fast nonlinear MPC was designed to optimize the waste heat recovery~\cite{WC19}. The multi-layer control designed in~\cite{DZ18} designed the MPC in the leader loop to optimize the thermal response to improved control performance. 
	
	The applicability of model predictive control expanded with the parametric solution of the MPC optimization problem, known as explicit MPC~\cite{Bemporad_automatica}. As the MPC optimization problem is pre-solved offline, it does not need to be solved in the online phase, i.e., in real-time control. Instead, a piece-wise affine (PWA) control law is evaluated to apply the optimal control action in each control step.
	The complexity of construction of the explicit MPC controller grows exponentially with the number of considered constraints. If the MPC design problem can be pre-solved explicitly offline, the consequent reduced online computational complexity makes the explicit MPC more suitable for practical industrial implementation. Nevertheless, the explicit MPC is not tunable in default as the conventional approach in~\cite{Bemporad_automatica} considers the penalty matrices with fixed structure and values. The inability to tune the explicit controller online can be a disadvantage due to varying operating conditions when the different setups of the controllers are beneficial.
	
	The possibility to tune the explicit MPC online came with the publishing of~\cite{Baric_tunable}. The tuning parameter penalizing the control inputs became a parameter, for which the optimal controller was precomputed. Nevertheless, the application was limited only to linear cost functions of the optimization problem. To satisfy the demands for often-used quadratic cost functions, the approximated tunable explicit MPC was presented in~\cite{Klauco_tunable}. The technique is based on two explicit model predictive controllers which differ in the setup of one penalty matrix. The two explicit MPCs provide upper and lower boundary optimal controllers. Based on the evaluation of the two boundary control inputs, the tuned control input is calculated by linear interpolation. The follow-up work~\cite{Oravec_tunable} provided stability and recursive feasibility guarantees by proper choice of the terminal penalty matrix and terminal set constraint~\cite{Mayne_stability}. Moreover, the strategy in~\cite{Oravec_tunable} extends the tuning ability based on \textit{any} penalty matrix and not just the input penalty.
	
	The idea of approximated tunable MPC with neural networks is presented in~\cite{Kis_NN_MPC}. To ensure the tuning property, the penalty matrices were included in the training process. As a result, it was possible to tune the neural network-based controller online, while mimicking the nearly optimal MPC. In~\cite{Kis_NN_MPC_corrector}, the neural network-based tunable controller MPC was extended with a corrector which steered the controller such that the constraints on the manipulated and process variables were satisfied. 
	
	The paper~\cite{self_tunable} pushes the idea of tunable explicit MPC further and deals with the issues of practical industrial-oriented implementation. In numerous practical applications, the reference value of the controlled variable is changed and acquires values from a wide range of operating conditions. The use of different controller setups can help handle the plant's nonlinear behavior. The paper~\cite{self_tunable} presents a procedure of the self-tunable controller technique. The controller's aggressivity is tuned based on the difference between the reference value and the steady state corresponding to the model linearization point. In the context of MPC, the aggressiveness is associated with the setup of the penalty matrices, as it determines the aggressiveness of the final control input. In general, higher penalization of the controlled states or control error in the cost function leads to more aggressive control actions. This process is analogous to increasing the proportional gain in the PID controller. On the contrary, higher penalization of the input variable leads to more sluggish control, e.g., see~\cite{Maciejowski_MPC}. In~\cite{self_tunable}, the MPC tuning based on the distance from the steady-state operating point represented a way how to compensate for the system's nonlinear behavior.  
	
	This work directly extends our results presented in~\cite{self_tunable}, where the basic principles of the self-tunable approximated explicit MPC were introduced. In this paper, a novel method of self-tuning parameter setup is introduced. Compared to~\cite{self_tunable}, the self-tuning method is based on the size of the reference step change. Moreover, the idea of further scaling of the tuning parameter is elaborated. The interval of the values of the self-tuning parameter is split at some certain value and each part of the interval corresponds to the specific operating conditions defined by the control engineer. In such a way, e.g., the system's asymmetric behavior is compensated. Finally, to investigate the benefits of the proposed approach, the proposed self-tuning control policy was implemented to control a laboratory-scaled counter-current plate heat exchanger. This work provides the control performance evaluation and analysis using the self-tunable controller compared to the boundary explicit MPCs.
	
	The paper is organized as follows. First, the theoretical backgrounds are presented in Section~\ref{sec:preliminaries}, where the explicit MPC, the approximated tunable explicit MPC, and existing self-tunable methods are briefly elaborated. Then, the novel proposed method of self-tunable procedure is explained in detail in Section~\ref{sec:methodology}. Finally, the experimental results of the self-tuning controller implementation on a heat exchanger are discussed in Section~\ref{sec:results}, followed by the main conclusions in Section~\ref{sec:conclusion}.
	
	\section{Theoretical backgrounds}
	\label{sec:preliminaries}
	
	In this section, the theoretical backgrounds necessary for the proposed method are summarized. First, the explicit model predictive control is briefly recalled. Next, the tunable technique of the approximated explicit model predictive control is introduced. Finally, the ideas of a self-tunable technique of the approximated explicit MPC are presented.
	
	\subsection{Explicit model predictive control}
	\label{sec:eMPC}
	
	Explicit model predictive control~\cite{Bemporad_automatica} utilizes a parametric solution of the model predictive control introducing its application range towards the systems with fast dynamics. Moreover, the explicit solution enables providing rigorous analysis and certification of the closed-loop system stability, constraints satisfaction, etc. As the explicit solution is available, real-time solving of the optimization problem in every control step is omitted. As this work deals with industrial-oriented implementation, let us consider the optimization problem in the following form:
	\begin{subequations}
		\label{eq:mpc_problem}
		\begin{eqnarray}
			\label{eq:mpc_problem_cost}
			\min_{u_0,u_{1},\ldots,u_{N-1}} &~& \! \sum_{k=0}^{N-1} \! \left( (y_k-y_\mathrm{ref})^{\intercal} Q_\mathrm{y} (y_k-y_\mathrm{ref}) + u_{k}^{\intercal} R u_{k} + x_{\mathrm{I},k}^{\intercal} Q_\mathrm{I} x_{\mathrm{I},k} \right)  \\
			\label{eq:mpc_problem_prediction_model_x}
			\mathrm{s.t.\!:} &~& \widetilde{x}_{k+1} = \widetilde{A}\,\widetilde{x}_{k} + \widetilde{B}\,u_{k}, \\
			\label{eq:mpc_problem_prediction_model_y}
			&~& y_{k} = \widetilde{C}\,\widetilde{x}_{k}, \\
			\label{eq:mpc_problem_input_constraints}
			&~& u_{k} \in \mathcal{U}, \\
			\label{eq:mpc_problem_state_constraints}
			&~& y_{k} \in \mathcal{Y}, \\
			\label{eq:mpc_problem_initial_coindition}
			&~& \widetilde{x}_{0} = \theta, \\
			\label{eq:mpc_problem_k_range}
			&~& k = 0,1,\ldots, N-1,
		\end{eqnarray}
	\end{subequations}
	where $k$ denotes the step of the prediction horizon $N$. 
	To obtain the offset-free control results, the built-in integrator was included in the state-space model, e.g., see~\cite{Ruscio_MPC_integral}. 
	The prediction model in Eq.~\eqref{eq:mpc_problem_prediction_model_x}--\eqref{eq:mpc_problem_prediction_model_y} has the form of augmented linear time-invariant (LTI) system for a given augmented state matrix $\widetilde{A} \in \mathbb{R}^{n_{\widetilde{\mathrm{x}}} \times n_{\widetilde{\mathrm{x}}}}$, augmented input matrix $\widetilde{B} \in \mathbb{R}^{n_{\widetilde{\mathrm{x}}} \times n_{\mathrm{u}}}$ and augmented output matrix $\widetilde{C} \in \mathbb{R}^{n_{\mathrm{y}} \times n_{\widetilde{\mathrm{x}}}}$. 
	Variables $\widetilde{x} \in \mathbb{R}^{n_{\widetilde{\mathrm{x}}}}$, $u \in \mathbb{R}^{n_{\mathrm{u}}}$, $y \in \mathbb{R}^{n_{\mathrm{y}}}$ are vectors of corresponding augmented system states, control inputs, and system outputs, respectively. 
	The sets $\mathcal{U} \subseteq \mathbb{R}^{n_{\mathrm{u}}}$, $\mathcal{Y} \subseteq \mathbb{R}^{n_{\mathrm{y}}}$ are convex polytopic sets of physical constraints on inputs and outputs, respectively. These sets include the origin in their strict interiors. The penalty matrix $Q_\mathrm{y} \in \mathbb{R}^{n_{\mathrm{y}} \times n_{\mathrm{y}}}, Q_\mathrm{y} \succeq 0$ penalizes the squared control error, i.e., the deviation between the controlled output and output reference value $y_\mathrm{ref}$. The matrix $R \in \mathbb{R}^{n_{\mathrm{u}} \times n_{\mathrm{u}}}, R \succ 0$ penalizes the squared value of control inputs. 
	The value of integrator is also penalized in the cost function with the penalty matrix $Q_\mathrm{I} \in \mathbb{R}^{n_{\mathrm{y}} \times n_{\mathrm{y}}}, Q_\mathrm{I} \succeq 0$. All the penalty matrices are considered to be diagonal due to the applicability of the self-tunable explicit MPC approach.
	The parameter $\theta \in \Theta$ in Eq.~\eqref{eq:mpc_problem_initial_coindition} represents the initial condition of the optimization problem for which it is parametrically pre-computed. 
	
	The augmented model of the controlled system with the built-in integrator in Eq.~\eqref{eq:mpc_problem_prediction_model_x}--\eqref{eq:mpc_problem_prediction_model_y} is rewritten as follows:
	\begin{subequations}
		\begin{eqnarray} 
			\label{eq:mpc_augmented_model_x} 
			\widetilde{x}_{k+1} &=& \begin{bmatrix} x_{k+1} \\ x_{\mathrm{I},k+1}\end{bmatrix} = \begin{bmatrix} A & \textit{0} \\ -T_\mathrm{s} C & I \end{bmatrix} \begin{bmatrix} x_{k} \\ x_{\mathrm{I},k} \end{bmatrix} + \begin{bmatrix} B \\ I \end{bmatrix} u_{k}, \\
			\label{eq:mpc_augmented_model_y}
			y_k &=& \begin{bmatrix} C & \textit{0} \end{bmatrix} \begin{bmatrix} x_{k} \\ x_{\mathrm{I},k} \end{bmatrix},
		\end{eqnarray}
	\end{subequations}
	where $x_{\mathrm{I}} \in \mathbb{R}^{n_{\mathrm{y}}}$ is the integral of the control error, $T_\mathrm{s}$ denotes the sampling time, and matrices $A$, $B$, $C$ are the well-known state-space matrices that form the augmented LTI model. As a consequence of this extension and penalization in the cost function in Eq.~\eqref{eq:mpc_problem_cost}, not only the control error is penalized, but also the integrated value, which leads to analogous offset-free reference tracking results as incorporating an integral part in the PID controller.
	
	The parametric solution of the optimization problem of the quadratic programming (QP) in Eq.~\eqref{eq:mpc_problem} leads to the explicit solution in the form of piecewise affine PWA control law defined above the domain consisting of $r$ critical regions:
	\begin{eqnarray}
		\label{eq:PWA_control_law}
		u(\theta) = \left\{ 
		\begin{matrix}
			F_{1} \, \theta + g_{1} & \mathrm{if} & \quad \theta \in \mathcal{R}_1, \\
			F_{2} \, \theta + g_{2} & \mathrm{else}\,\,\mathrm{if} &\quad \theta \in \mathcal{R}_2, \\
			& \vdots & \\
			F_{r} \, \theta + g_{r} & \mathrm{else}\,\,\mathrm{if} & \quad \theta \in \mathcal{R}_{r}, \\
		\end{matrix}
		\right.
	\end{eqnarray}
	where $F_{i} \in \mathbb{R}^{n_{\mathrm{u}} \times n_{\mathrm{x}}}$ and $g_{i}  \in \mathbb{R}^{n_{\mathrm{u}}}$ respectively are the slope and affine section of the corresponding control law. The PWA function defined in Eq.~\eqref{eq:PWA_control_law} is stored and recalled in the online phase, i.e., during the real-time control. Based on identifying the specific polytopic critical region $\mathcal{R}_{i}$, where the parameter $\theta$ belongs, the optimal control input is calculated based on the associated control law in Eq.~\eqref{eq:PWA_control_law}.
	
	Note, many other formulations of the optimization problems for the explicit MPC design were formulated mainly in terms of the definition of the cost functions in Eq.~\eqref{eq:mpc_problem_cost}. Also, the incremental (velocity) formulation of the state-space model is common, but leads to further extension of the vector of parameters $\theta$, and therefore also the complexity of the explicit MPC controller increases. Another option for offset-free tracking is introducing disturbance modeling and estimation. For such an overview see, e.g,.~\cite{Klauco_mpc}

	\subsection{Tunable explicit model predictive control}
	\label{sec:tunable}
	
	The aggressivity of the controller and the whole nature of the control is influenced by appropriate fine-tuning of the penalty matrices in the optimization problem in Eq.~\eqref{eq:mpc_problem}. When the multi-parametric QP (mp-QP) problem is precomputed offline to obtain the corresponding parametric solution, it is not possible to tune the controller afterward without trading off a significant increase in the controller complexity or the performance loss. As the operating conditions and requirements on controller setup may differ throughout the control, the ability to adjust the controller's aggressivity can be very beneficial.
	
	The idea of approximated tunable explicit MPC comes from the work ~\cite{Klauco_tunable}, where the control action is calculated based on linear interpolation between two boundary control actions. These control actions result from evaluating two boundary explicit MPCs. The boundary explicit controllers are constructed by solving the optimization problem having the same structure and setup, except for one of the penalty matrices -- the tuned one. Based on the specific control application, any penalty matrix can be chosen as the tuned parameter, i.e., this approach is applicable for any penalty matrix. The boundary penalty matrices follow the assumptions on the penalty matrices from Section 2.1 and are diagonal matrices such that $\lambda_{i,\mathrm{L}} \le \lambda_{i,\mathrm{U}}$, $\forall i = 1,\dots,s$, where $\lambda$ denotes the vector of eigenvalues of the penalty matrix, $s$ is the rank of the tuned penalty matrix, and $L$, $U$ denote the lower and upper boundary setup, respectively. 
	
	Let us consider the penalty matrices in the cost function in Eq.~\eqref{eq:mpc_problem_cost}. The penalty matrices are scaled in the following way:
	\begin{subequations}
		\label{eq:tunable_matrices}
		\begin{eqnarray}
			\label{eq:tunable_R}
			&~& R(k) = (1-\rho(k)) \, R_\mathrm{L} + \rho(k) \, R_\mathrm{U}, \\
			\label{eq:tunable_Qx}
			&~& Q_\mathrm{I}(k) = (1-\rho(k)) \, Q_\mathrm{I,L} + \rho(k) \, Q_\mathrm{I,U}, \\
			\label{eq:tunable_Qy}
			&~& Q_\mathrm{y}(k) = (1-\rho(k)) \, Q_\mathrm{y,L} + \rho(k) \, Q_\mathrm{y,U},
		\end{eqnarray}
	\end{subequations}
	where $\rho$ represents the tuning parameter such that $0 \le \rho \le 1$ holds. Based on the rules in Eq.~\eqref{eq:tunable_matrices}, it is possible to choose online any controller setup from the lower to the upper boundary of the tuned matrix. From the implementation point of view, it is preferred to tune just a single penalty matrix, i.e., to store only two controllers corresponding to the boundary values of the selected penalty matrix. To determine which penalty matrix in Eq.~\eqref{eq:tunable_matrices} should be tuned, it is suggested to judge the control performance by systematic tuning of all the penalty matrices. Systematic tuning involves selecting a specific penalty matrix and observing the control results by gradually increasing or decreasing the diagonal elements of the matrix. This process is then repeated for the remaining penalty matrices in a similar manner.
	
	When the tuning parameter $\rho$ is determined based on the current control conditions, the approximated optimal control action is evaluated using the two optimal controllers. Based on the boundary control actions, the interpolated, i.e., tuned control action is calculated using the convex combination:
		\begin{eqnarray}
			\label{eq:tunable_u}
			u(k) = (1-\rho(k)) \, u_\mathrm{L}(k) + \rho(k) \, u_\mathrm{U}(k),
		\end{eqnarray}
	where $u_\mathrm{L}$ and $u_\mathrm{U}$ denote the optimal control actions from the lower and upper boundary controller, respectively. The online tuning of the controller comes with the cost of storing and evaluating two explicit controllers. Nevertheless, the ability to tune the controller may be more important in many practical applications.
	
	The concept of explicit MPC tuning is applicable to a wide class of MPC design formulations, based on the current specific needs. Without loss of generality, hereafter, let us consider the penalty matrices of the cost function in Eq.~\eqref{eq:mpc_problem_cost}, as it is necessary to satisfy offset-free reference tracking.
	
	\begin{remark}
		If the asymptotic stability and recursive feasibility guarantees are required, the reader is referred to follow the instructions from \cite{Oravec_tunable}. In order to satisfy these requirements, the study introduces a procedure for computing the common terminal penalty and terminal set for the two boundary controllers. 
	\end{remark}
	
	\begin{remark}
		Not only Eq.~\eqref{eq:tunable_u} needs to be chosen for interpolation of the control input. Another way of tuning of the control input can be using some nonlinear relation for the interpolation. 
	\end{remark}

	\subsection{Self-tunable explicit model predictive control}
	\label{sec:self_tunable}	
	The advantage of a tunable controller brings a question of how to design the logic of setting the tuning parameter $\rho$. In this section, the idea of online self-tuning is summarized~\cite{self_tunable}. The concept of self-tuning provides the possibility to adjust the aggressiveness of the controller without the necessity to intervene and tune the penalty matrices during control. 
	
	The need for real-time controller tuning often arises from tracking a time-varying piece-wise constant (PWC) reference. The work~\cite{self_tunable} focuses on adjusting the penalty matrix when the reference value is changed. The further the reference value is from the steady state, the more aggressively the controller is tuned. The idea behind the suggested scaling lies in compensation for the nonlinear behavior of the system.  
	
	Consider a single-input and single-output (SISO) system or a system with completely decoupled pairs of the control inputs and the system outputs. 
	Then, the procedure of tuning the controller is based on evaluating the different operating points between the current value of the reference and the system steady-state value. This deviation is considered to scale the value of control action. First, the maximal admissible absolute value of the reference is defined. Analogous to the reference trajectory preview concept of MPC design, this value can be determined based on the general knowledge of the expected future reference values. Another suggestion is to set the maximal deviation $d_{\max}$ based on the constraints on system outputs: 	
	\begin{eqnarray}
		\label{eq:d_max}
		d_{\max} = \max(\vert y_{\min} \vert, y_{\max}),
	\end{eqnarray}	
	where the symbol $\lvert.\rvert$, hereafter, denotes the element-wise absolute value, $y_{\min}$ and $y_{\max}$ are respectively lower and upper bound on the output variable in the deviation form, i.e., zero (origin) corresponds to the system steady-state value. Using the information about the maximal possible deviation $d_{\max}$, the tuning parameter $\rho$ can be calculated as the ratio between the current reference value and the maximal deviation: 
	\begin{eqnarray}
		\label{eq:rho_self_tunable}
		\rho (k) = \frac{\vert y_{\mathrm{ref}}(k) \vert}{d_{\max}}.
	\end{eqnarray}
	Based on Eq.~\eqref{eq:rho_self_tunable}, the property $0 \le \rho \le 1$ holds, as $\vert y_{\mathrm{ref}} \vert \le d_{\max}$. As a consequence, the parameter $\rho$ represents a way how to normalize the deviation from the steady-state value and is exploited to scale the control action or, implicitly, to tune the aggressiveness of the controller. 
	
	Note that the reference value must be reachable from the operating range to ensure that $0 \le \rho \le 1 $ holds. Otherwise, the interpolated control action would be the ``extrapolation'' leading to the loss of guarantees on the input or state constraints satisfaction, etc.
	
	When considering tuning the control action based on Eq.~\eqref{eq:tunable_u}, a higher value of tuning parameter $\rho$ leads to approaching the upper boundary controller and vice versa. When tuning, e.g., the matrix $Q_\mathrm{y}$ penalizing the control error, a higher ratio $\rho$ would lead to more aggressive control actions. When operating with the reference value close to the system steady-state value, the parameter $\rho$ decreases and the control profiles become sluggish.
	
	\begin{remark}
		In general, the parameter $d_{\max}$ is a vector, as it depends on the size of the system outputs. If $d_{\max}$ is scalar, the parameter $\rho$ is scalar as well and can be directly utilized to scale the control action. If multiple outputs are controlled, it is suggested to calculate the tuning parameter based on the maximal ratio as follows:
		\begin{eqnarray}
			\label{eq:rho_max}
			\rho(k) = \max \left( \frac{\vert y_{\mathrm{ref}}(k) \vert}{d_{\max}} \right).
		\end{eqnarray}
	\end{remark}
	
	Note that the relations in Eq.~\eqref{eq:rho_self_tunable} and Eq.~\eqref{eq:rho_max} operate with the absolute value of the reference. It is not taken into account whether the reference value changed upwards or downwards with respect to the system steady-state value placed in the origin, i.e., whether the inequality $\Delta_\mathrm{ref}(k) = y_\mathrm{ref}(k)-y_\mathrm{ref}(k-1) > 0$ holds or $\Delta_\mathrm{ref}(k) < 0$. As many plants have nonlinear behavior with an asymmetric nature (different behavior when the process variable is rising or decreasing), the positivity or negativity of the reference change could be considered in the controller self-tuning procedure to improve the control performance.

	\section{Methodology}
	\label{sec:methodology}
	This section extends the ideas of self-tunable explicit MPC in order to improve control performance. First, a different way of tuning parameter calculation is introduced. Furthermore, an extended self-tunable technique is presented to scale the tuning parameter for industrial-oriented applications, when it is beneficial to exploit a specific range of the tuning parameter in different operating conditions.

	\subsection{Tuning parameter based on the size of reference change}
	\label{sec:self_tunable_delta_ref}
	The approach of self-tunable explicit MPC in~\cite{self_tunable} suggested tuning based on the current reference value distance from the steady state. The aim is to compensate for the nonlinear behavior of a system when using a simple linear prediction model. This work provides also another useful way of the real-time evaluation of the tuning parameter $\rho$ based on the size of reference change. When different sizes of reference step changes are made and the behavior of the closed-loop system is varying, it can be beneficial to include the size of the reference step change in the tuning procedure.
	
	In this approach, the aggressivity is adjusted based on the ratio between the reference step change and the maximal reference step change that can be realized during the control operation:
	\begin{eqnarray}
		\label{eq:rho_delta_ref}
		\rho(k) = \frac{\vert \Delta_{\mathrm{ref}}(k) \vert}{\Delta_{\max}},
	\end{eqnarray} 
	where $\Delta_{\mathrm{ref}}(k) = y_{\mathrm{ref}}(k) - y_{\mathrm{ref}}(k-1)$ is the size of the reference step change. The denominator of Eq.~\eqref{eq:rho_delta_ref} is changed as well. In contrast to the maximal deviation from the steady state in Section~\ref{sec:self_tunable}, this approach introduces $\Delta_{\max}$ as the maximal possible reference step change. Analogously to the original approach, the maximal reference step can be set based on the general knowledge of the expected future reference values, i.e., $\Delta_{\max} = \Vert \Delta_{\mathrm{ref}}(k) \Vert_{\infty}, \forall k \ge 0 $. 
	Another option is to exploit the information about the system constraints and set the parameter $\Delta_{\max}$ according to Eq.~\eqref{eq:d_max}. 
	
	Note, only the absolute value of $\Delta_{\mathrm{ref}}$ and $\Delta_{\max}$ are considered in this procedure to ensure $\rho \ge 0$. 
	
	In Eq.~\eqref{eq:rho_delta_ref}, it is suggested to increase the value of tuning parameter $\rho$ with increasing value of reference step change. Note, in this work, the larger value of the tuning parameter leads to adding more weight on the penalty matrices associated with the upper boundary controller, see Eq.~\eqref{eq:tunable_matrices}. If the opposite logic of controller tuning is requested, it is possible to adapt the tuning such that 
	\begin{eqnarray}
		\label{eq:tunable_R_2}
		&~& R(k) = \rho(k) \, R_\mathrm{L} + (1-\rho(k)) \, R_\mathrm{U}, \\
		\label{eq:tunable_Qx_2}
		&~& Q_\mathrm{I}(k) = \rho(k) \, Q_\mathrm{I,L} + (1-\rho(k)) \, Q_\mathrm{I,U}, \\
		\label{eq:tunable_Qy_2}
		&~& Q_\mathrm{y}(k) = \rho(k) \, Q_\mathrm{y,L} + (1-\rho(k)) \, Q_\mathrm{y,U},
	\end{eqnarray}
	hold. This change leads to adding more weight to the lower boundary controller with the increasing value of the tuning parameter $\rho$.
	
	\begin{remark}
		The tuning parameter $\rho$ should be updated only when the reference changes. Updating the tuning parameter in the control steps when $\Delta_{\mathrm{ref}} = 0$ would lead to using tuning parameter $\rho$ with zero value, i.e., the control input would correspond to one boundary controller and would not be scaled.
	\end{remark}

	\subsection{Self-tunable technique for systems with asymmetric behavior}
	\label{sec:self_tunable_rho_scaling}
	
	This paper provides a further extension of the self-tuning method proposed in~\cite{self_tunable}. The suggested technique of tuning is suitable, e.g., for systems with asymmetric behavior, but can be used in any application, where ``simple'' tuning in the whole range of tuning parameter $\rho$ is not sufficient.
	
	The proposed self-tuning method is based on splitting the interval of the tuning parameter $\rho$ in order to utilize different parts of the interval in different operating conditions. Instead of the original value of tuning parameter $\rho$, the adjusted tuning parameter $\widetilde{\rho}$ is then utilized to scale the control input according to Eq.~\eqref{eq:tunable_u}.
	
	\begin{definition}[Decision function]
		\label{def:gamma}
		For a given interval of tuning parameter $\rho$, $0 \leq \rho \leq 1$, let $\rho_{\mathrm{s}}$, $0 < \rho_{\mathrm{s}} < 1$ be a boundary value splitting the interval into two parts. Let $\gamma: \mathbb{R} \rightarrow \mathbb{R}$ be an arbitrary function such that $0 \leq \gamma \leq 1$ holds. Then the decision function $\gamma$ is constructed to assign its value either $\gamma \leq \rho_{\mathrm{s}}$ or $\gamma \ge \rho_{\mathrm{s}}$.
	\end{definition}
	Various decision functions $\gamma$ can be considered. In this work, the decision functions according to Eq.~\eqref{eq:rho_max} and Eq.~\eqref{eq:rho_delta_ref} are suggested, while Eq.~\eqref{eq:rho_delta_ref} was implemented in the experimental case study.
	
	\begin{definition}[Scaling of the tuning parameter]
		\label{def:rho_tilde}
		Given the value of tuning parameter $\rho$, $0 \leq \rho \leq 1$, the splitting value of the tuning parameter interval $\rho_{\mathrm{s}}$, $0 < \rho_{\mathrm{s}} < 1$, and the value of the decision function $\gamma$, $0 \leq \gamma \leq 1$. Then the scaling of the tuning parameter $\widetilde{\rho}$ is given by:		
		\begin{eqnarray}
			\label{eq:rho_tilde_def}
			\widetilde{\rho} = \left\{ 
			\begin{matrix}
				\rho \, \rho_{\mathrm{s}} & \mathrm{if} & \quad \gamma \in \langle 0, \rho_{\mathrm{s}} \rangle , \\
				\rho \, (1-\rho_{\mathrm{s}}) + \rho_{\mathrm{s}}, & \mathrm{else}\,\,\mathrm{if} &\quad \gamma \in \langle \rho_{\mathrm{s}}, 1 \rangle.
			\end{matrix}
			\right.
		\end{eqnarray}
	\end{definition}
	
	
	
	\begin{remark}
		\label{rem:rho_tilde}		
		The introduction of splitting the tuning parameter $\widetilde{\rho}$ into the tuning intervals in~\eqref{eq:rho_tilde_def} is not limited only to two intervals. If the nature of the controlled plant would benefit from splitting the operating range into more intervals, e.g., when the plant operates in the multiple steady-states values, then these intervals are simply determined by the corresponding values of $\rho_{\mathrm{s}, i}$ for each part of the interval. Next, the tuning rules in~\eqref{eq:rho_tilde_def} are adopted in an analogous way.
	\end{remark}
	
	The following outcomes result from Eq.~\eqref{eq:rho_tilde_def}. 
	%
	\begin{lemma}
		\label{lem:PWA_control_law_interval}
		Given control law in~\eqref{eq:PWA_control_law}, its approximation given by the convex combination in~\eqref{eq:tunable_u}, and given scaled tuning parameter $\widetilde{\rho}$ according to Definition~\ref{def:rho_tilde}. Then the control action approximated into the form:
		\begin{eqnarray}
			\label{eq:PWA_control_law_interval}
			u(k) = (1-\widetilde{\rho}(k)) \, u_\mathrm{L}(k) + \widetilde{\rho}(k) \, u_\mathrm{U}(k),
		\end{eqnarray}
		preserves the closed-loop system stability and recursive feasibility of the original control law in~\eqref{eq:PWA_control_law}.
	\end{lemma}
	
	\begin{proof}
		\label{proof:rho_tilde}
		It has been proven~\cite{Oravec_tunable} that for the asymptotic stable and recursive feasible pair of control inputs $(u_{\mathrm{L}}, u_{\mathrm{U}})$, the approximated control law in~\eqref{eq:PWA_control_law} preserves these properties for any $\rho$ satisfying $0 \leq \rho \leq 1$, see Theorem~3.6 in~\cite{Oravec_tunable}. 
		It remains to prove that for any value of the scaled tuning parameter $\widetilde{\rho}$ according to the Definition~\ref{def:rho_tilde} the same results hold. 
		The rest of the proof of Lemma~\ref{lem:PWA_control_law_interval} consists of two parts corresponding to each particular rule in~\eqref{eq:rho_tilde_def}. 
		\\
		First, it is proved that the Lemma~\ref{lem:PWA_control_law_interval} holds for any $\gamma \leq \rho_{\mathrm{s}}$. Substituting a lower bound $\rho=0$ into~\eqref{eq:rho_tilde_def} leads to
		$\widetilde{\rho} = 0$. 
		For the upper bound value of $\rho=1$, from~\eqref{eq:rho_tilde_def} holds $\widetilde{\rho} = \rho_{\mathrm{s}} < 1$. Next, for any value $0 < \rho < 1$ evaluation of the linear rule in~\eqref{eq:rho_tilde_def} leads to the convex combination, i.e., $0 < \widetilde{\rho} < \rho_{\mathrm{s}} $ holds. 
		Therefore, any value of $\widetilde{\rho}$ satisfies $0 \leq \widetilde{\rho} \leq \rho_{\mathrm{s}} < 1$. As a consequence, according to the Theorem~3.6 in~\cite{Oravec_tunable}, the asymptotic stability and recursive feasibility of the control law in~\eqref{eq:PWA_control_law_interval} are preserved. 
		\\
		Secondly, it is proved that the Lemma~\ref{lem:PWA_control_law_interval} holds also for any $\gamma \geq \rho_{\mathrm{s}}$. Substituting a lower bound $\rho=0$ into~\eqref{eq:rho_tilde_def} leads to
		$\widetilde{\rho} = \rho_{\mathrm{s}}$. 
		For the upper bound value of $\rho=1$, from~\eqref{eq:rho_tilde_def} holds $\widetilde{\rho} = 1$. Next, for any value $0 < \rho < 1$ evaluation of the linear rule in~\eqref{eq:rho_tilde_def} leads to the convex combination, i.e., $\rho_{\mathrm{s}} < \widetilde{\rho} < 1 $ holds. 
		Therefore, any value of $\widetilde{\rho}$ satisfies $\rho_{\mathrm{s}} \leq \widetilde{\rho} \leq 1$. As a consequence, according to the Theorem~3.6 in~\cite{Oravec_tunable}, the asymptotic stability and recursive feasibility of the control law in~\eqref{eq:PWA_control_law_interval} are preserved. 
	\end{proof}
	
	\begin{remark}
		The Lemma~\ref{lem:PWA_control_law_interval} can be extended subject to the multiple intervals in an analogous way following the Remark~\ref{rem:rho_tilde}.
	\end{remark}
	
	The advantage of the proposed method remains in the self-tuning of the controller as in the approach from Section~\ref{sec:self_tunable}. Nevertheless, it is required to appropriately determine the splitting value of the tuning parameter $\rho_{\mathrm{s}}$ and assign the parts of the interval to the associated operating conditions.
	
	\begin{remark}
		Note, the suggested scaling method is suitable also for online MPC, as the optimization problem is solved in every control step. Therefore, it is possible to include the controller tuning in the procedure of computing the optimal control input.   
	\end{remark}
	
	For a detailed insight into the proposed control technique, the procedure of self-tuning evaluation is depicted in Figure~\ref{fig:tuning}. 
	
	From the point of computational complexity, the proposed tuning procedure does not lead to any significantly demanding mathematical operations. Simple algebraic operations in Eq.~\eqref{eq:rho_delta_ref} and Eq.~\eqref{eq:rho_tilde_def} are evaluated. Note, the overall control strategy still comes with the cost of storing and evaluating two explicit controllers.
	
	\begin{figure}
		\begin{center}
			\includegraphics[width=\textwidth]{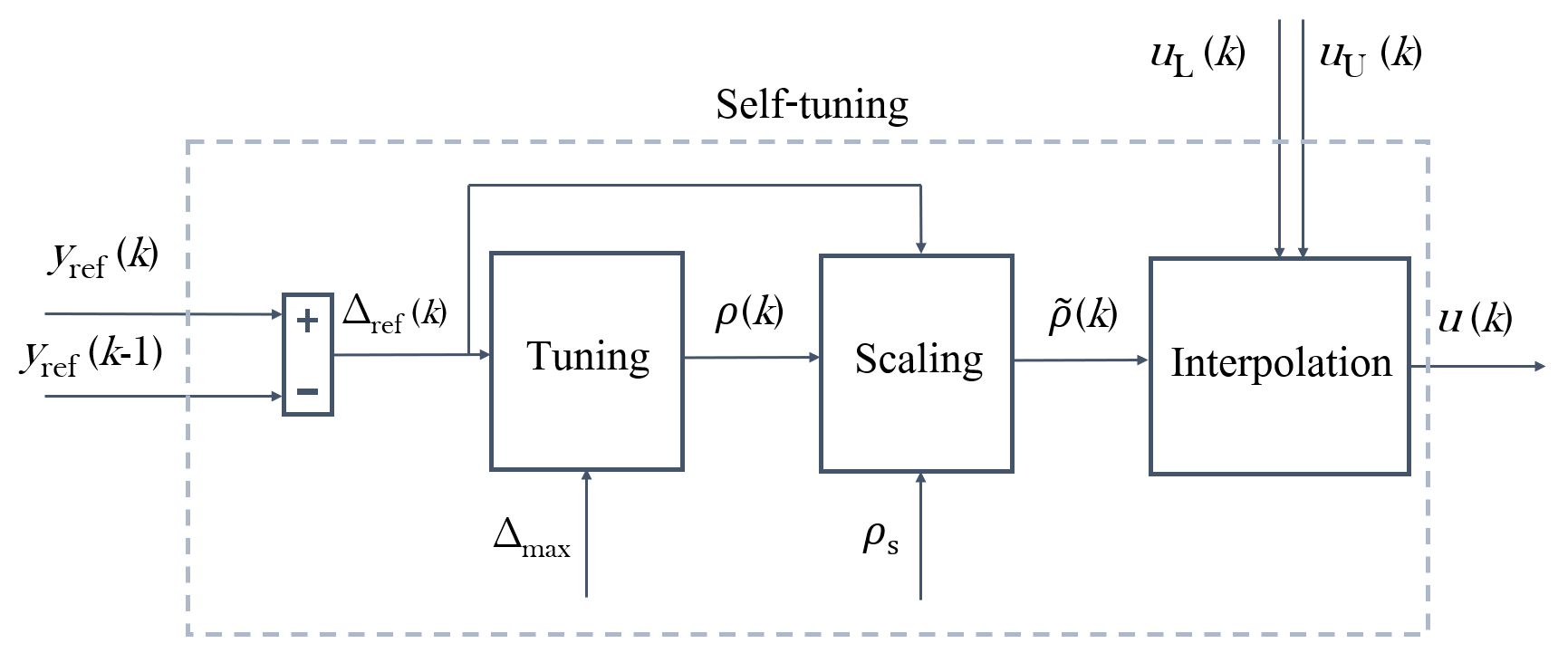}
			\caption[Self-tuning scheme]{Scheme of the self-tuning control evaluation.}
			\label{fig:tuning}
		\end{center}
	\end{figure}

	\section{Results and discussion}
	\label{sec:results}
	
	In this section, the results of the proposed self-tuning method are analyzed by an experimental implementation. The self-tuning strategy utilizes tuning parameter calculation based on the size of reference change (Section~\ref{sec:self_tunable_delta_ref}) and the scaling of tuning parameter based on splitting the interval of the parameter and assigning the interval parts to specific operating conditions (Section~\ref{sec:self_tunable_rho_scaling}).

	The plant on which the control was implemented and analyzed is a laboratory-scaled counter-current liquid-to-liquid plate heat exchanger Armfield Process Plant Trainer PCT23~\cite{pct23}, see Figure~\ref{fig:HE}. The schematic of the plant is depicted in Figure~\ref{fig:HE_scheme}. The heat exchanger is 90\,mm wide, 103\,mm long, and 160\,mm high. The heat exchange is performed between the cold medium (water) and the hot medium (water). The cold medium as well as the heating medium are transported to the heat exchanger by two peristaltic pumps with flexible tubing from silicon rubber. The flow rate of the cold medium is constant, while the aim of control is to track the reference value of the outlet cold medium temperature. Therefore, the controlled variable is the cold medium temperature $T$ at the outlet of the heat exchanger. The inlet cold medium temperature was constant during the whole control, i.e., $T_\mathrm{C} = 19^{\circ}$\,C. The temperature of the heated cold medium in the outlet stream was measured by the type K thermocouple. The associated manipulated variable is the voltage $U$ corresponding to the power of the pump feeding the heat exchanger by the hot medium. The voltage is within the range of [$0-5$]\,V normalized into the relative values in percentage. The maximal voltage 5\,V or 100\% corresponds to volumetric flow rate 11.5\,$\mathrm{ml}\,\mathrm{s}^{-1}$. For further technical specifications of the laboratory heat exchanger, the reader is referred to~\cite{pct23}. As heat exchange is a nonlinear and asymmetric process~\cite{Liptak}, this heat exchanger represents a suitable candidate for the presented controller tuning strategy. The corresponding illustrative scheme of the implemented closed-loop control setup is in Figure~\ref{fig:CL}, where the ``Self-tuning'' block substitutes the more detailed scheme of the tuning procedure in Figure~\ref{fig:tuning}.  

\begin{figure}
	\begin{center}
		\includegraphics[width=0.8\textwidth]{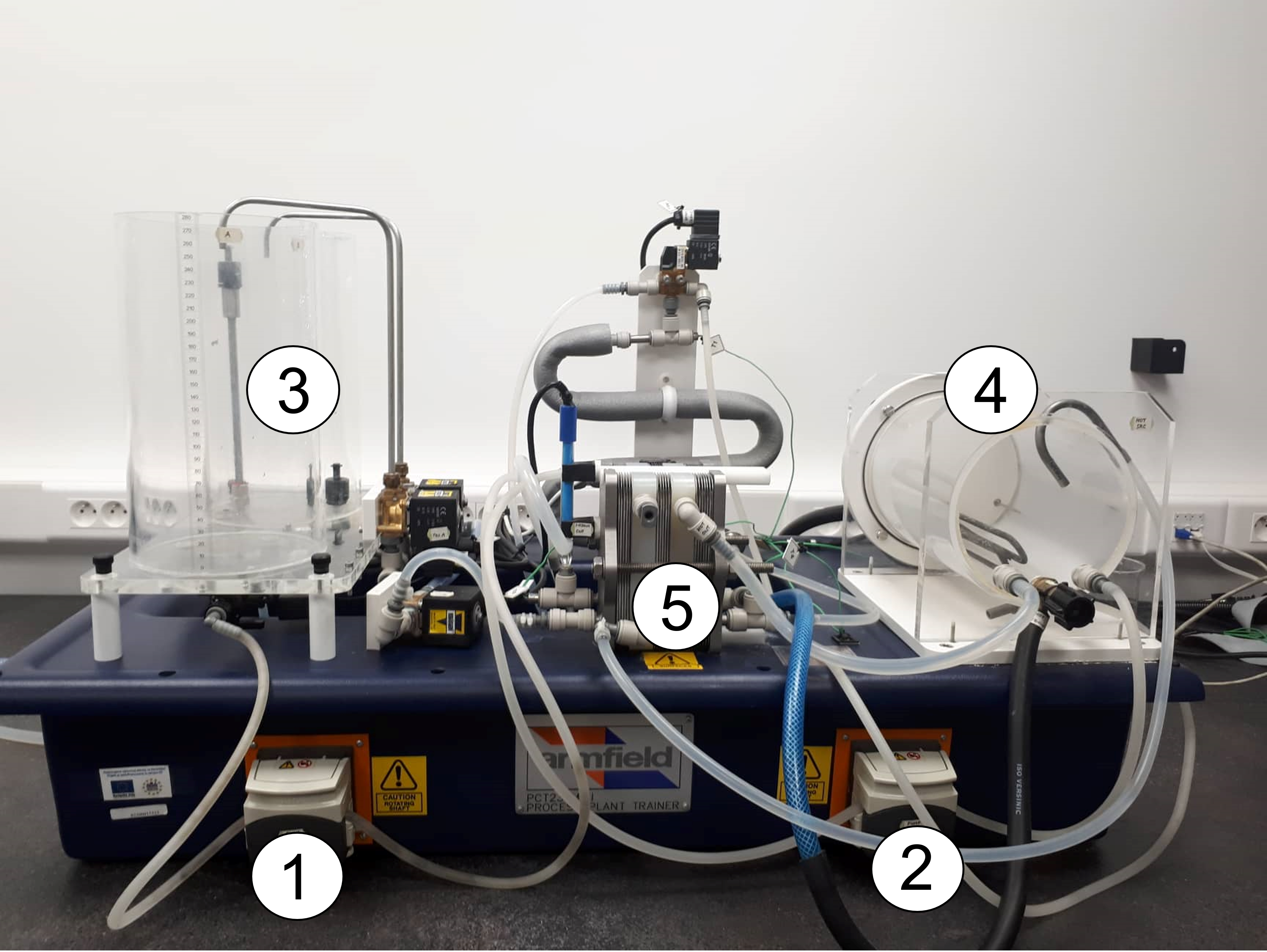}
		\caption[Heat exchanger Armfield Process Plant Trainer PCT23]{Laboratory heat exchanger Armfield Process Plant Trainer PCT23: cold medium pump\,(1), heating medium pump\,(2), cold medium tanks\,(3), heater for heating medium\,(4), heat exchanger\,(5).}
		\label{fig:HE}
	\end{center}
\end{figure}

\begin{figure}
	\begin{center}
		\includegraphics[width=0.8\textwidth]{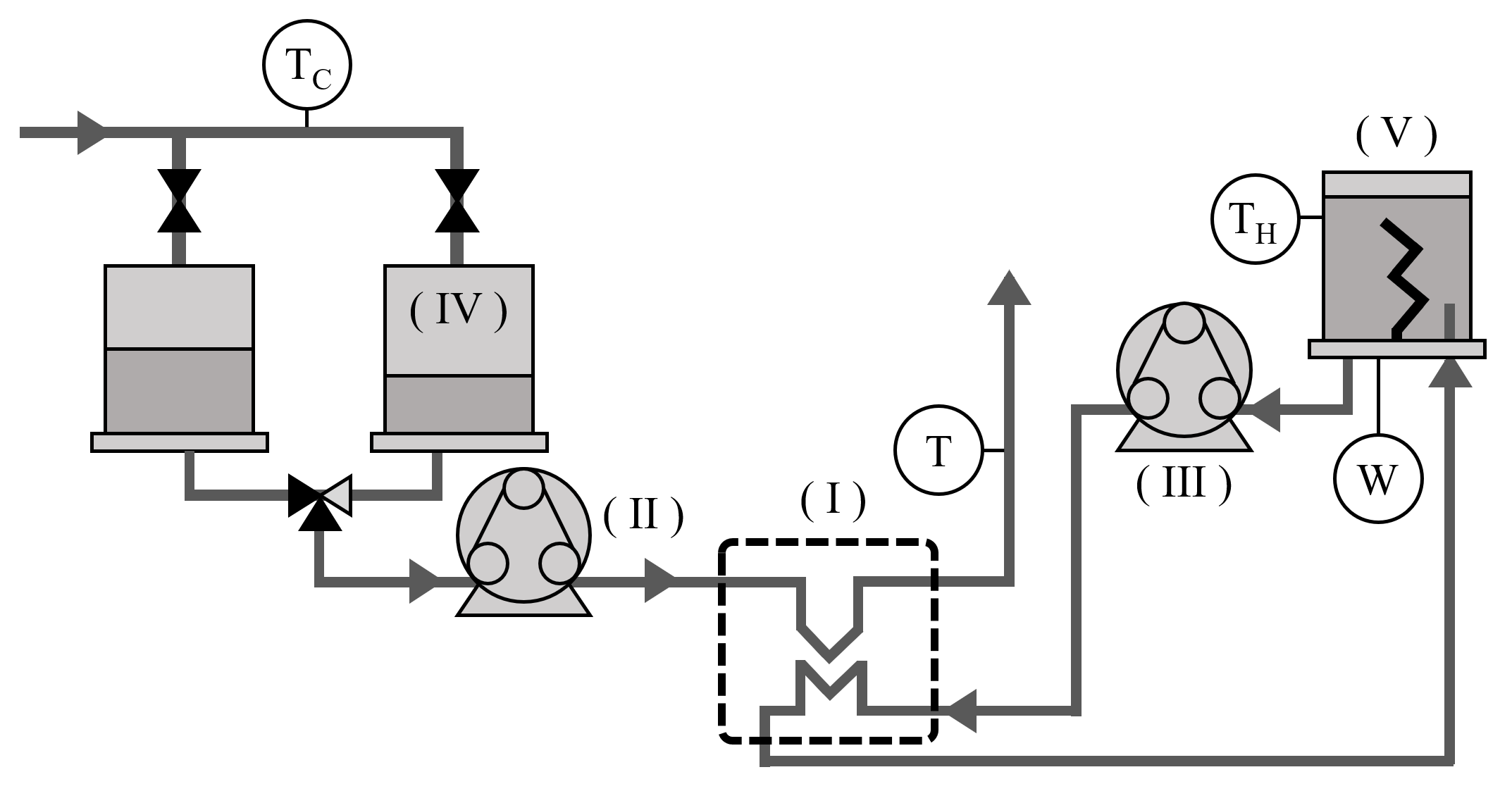}
		\caption{Scheme of Armfield PCT23. Heat exchanger (I), peristaltic pump for cold medium (II), peristaltic pump for heating medium (III), tank for cold medium (IV), heater for heating medium (V), temperature sensors ($\mathrm{T}$ -- controlled temperature, $\mathrm{T}_\mathrm{C}$ -- cold outlet cold medium temperature, $\mathrm{T}_\mathrm{H}$ -- heating medium temperature), and electric power for maintaining the temperature of the heating medium (W). }
		\label{fig:HE_scheme}
	\end{center}
\end{figure}

The system was identified by experimental identification. The aim was to work with linear nominal model in MPC optimization problem to decrease the numerical complexity. To avoid plant-model mismatch in order to ensure offset-free tracking, either disturbance observer or built-in integrator (Eq.~\eqref{eq:mpc_problem}) can be employed. Due to the ease of implementation, in this work, the built-in integrator was considered. The system was identified based on several measured step responses. The step changes were performed in the whole range of admissible values of manipulated variable and every step response was identified by transfer function. It was possible to identify every step response as a first-order system, while the nominal gain and time constant are respectively $K = 0.24\,^{\circ}\mathrm{C}$ and $\tau = 5.7$\,s. 
Finally, the nominal transfer function was converted to the state-space model. The matrices of the discrete-time state-space model of the plant are
\begin{subequations}
	\label{eq:model_A_B} 
	\begin{eqnarray}
		A = \begin{bmatrix}
			0.839
		\end{bmatrix}, \quad
		B = \begin{bmatrix}
			0.039
		\end{bmatrix}, \quad
		C = \begin{bmatrix}
			1
		\end{bmatrix}, 
	\end{eqnarray}
\end{subequations}
considering the sampling time $T_\mathrm{s}$ = 1\,s. 
To respect the physical limitations of the operating conditions, the constraints are considered in the terms of control inputs
\begin{eqnarray}
	\label{eq:u_const}
	-15\,\% \le u \le 65\,\%,
\end{eqnarray}
where the variable $u$ represents the control inputs in the deviation form. The values of the heated cold medium temperature and voltage of the heating medium pump corresponding to zero steady states are respectively $T^\mathrm{s}$~=~35\,$^{\circ}\mathrm{C}$ and $U^\mathrm{s}$~=~35\,\%. Therefore, the physical constraints on the manipulated variable are actually
\begin{eqnarray}
	\label{eq:U_const}
	20\,\% \le U \le 100\,\%.
\end{eqnarray}	

As the controlled system is naturally stable even if the maximal or minimal value of the manipulated variable is constantly applied, the constraints on the controlled variable in Eq.~\eqref{eq:mpc_problem_state_constraints} could be omitted. On the other hand, unbounded states/outputs lead to higher memory consumption, because covering the whole possible range of parameters requires more critical regions. Therefore, the ``redundant'' constraints on the system outputs were included in order to reduce the number of critical regions and the overall memory footprint of the explicit controllers. The output constraints were set as: 
\begin{eqnarray}
	\label{eq:y_const}
	-15\,^{\circ}\mathrm{C} \le y \le 20\,^{\circ}\mathrm{C}.
\end{eqnarray}	

The constraints in Eq.~\eqref{eq:y_const} are equal to physical temperature as follows:
\begin{eqnarray}
	\label{eq:Y_const}
	20\,^{\circ}\mathrm{C} \le y \le 55\,^{\circ}\mathrm{C},
\end{eqnarray}
which corresponds to the range of temperature values which are achievable in the considered laboratory conditions and setup.		 

The penalty matrices of the problem in Eq.~\eqref{eq:mpc_problem} were systematically tuned, and the corresponding control setup was implemented on the laboratory heat exchanger for each setup of the considered explicit MPC controllers. 
First, the tuning procedure aimed to determine which penalty matrix is the most suitable for real-time tuning. The most relevant was the penalty matrix $Q_\mathrm{y}$ as the tuning is directly associated with a reference value, which takes place in the calculation of the tuning factor $\rho$. Moreover, the tuning of $Q_\mathrm{y}$ preserved a satisfactory control performance. Next, the boundary values of the tunable matrix $Q_\mathrm{y}$ were tuned until the following limit values were determined based on the measured closed-loop control data: $Q_\mathrm{y, L}$ = 100 and $Q_\mathrm{y, U}$ = 1\,000. The built-in integrator was penalized with the fixed penalty matrix $Q_\mathrm{I}$ = 1 and the control input with the fixed penalty matrix $R$ = 10. The prediction horizon $N$ was set to 20 control steps. The explicit model predictive controllers were constructed in MATLAB R2020b using the Multi-Parametric Toolbox 3~\cite{mpt_conf}. 

The explicit MPC corresponding to the penalty matrix $Q_\mathrm{y, U}$ contains 1\,680 critical regions, and the explicit MPC with the penalty matrix $Q_\mathrm{y, L}$ contains 409 critical regions. The corresponding polytopic partitions can be seen in Figure~\ref{fig:partition_U} for the upper boundary controller and Figure~\ref{fig:partition_L} for the lower boundary controller.  

\begin{figure}
	\begin{center}
		\includegraphics[width=0.8\textwidth]{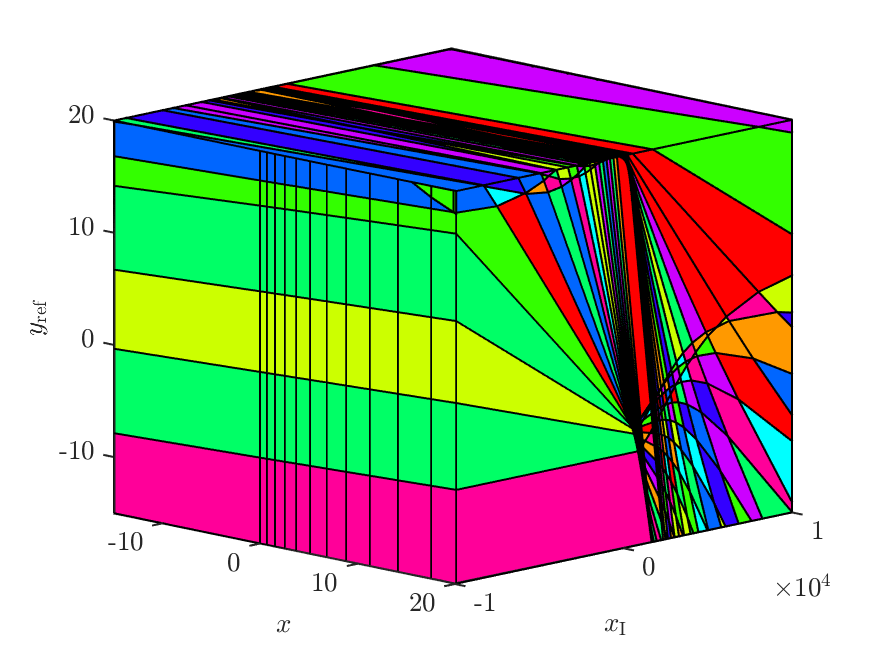}
		\caption{Polytopic partition of the upper boundary explicit MPC.}
		\label{fig:partition_U}
	\end{center}
\end{figure}

\begin{figure}
	\begin{center}
		\includegraphics[width=0.8\textwidth]{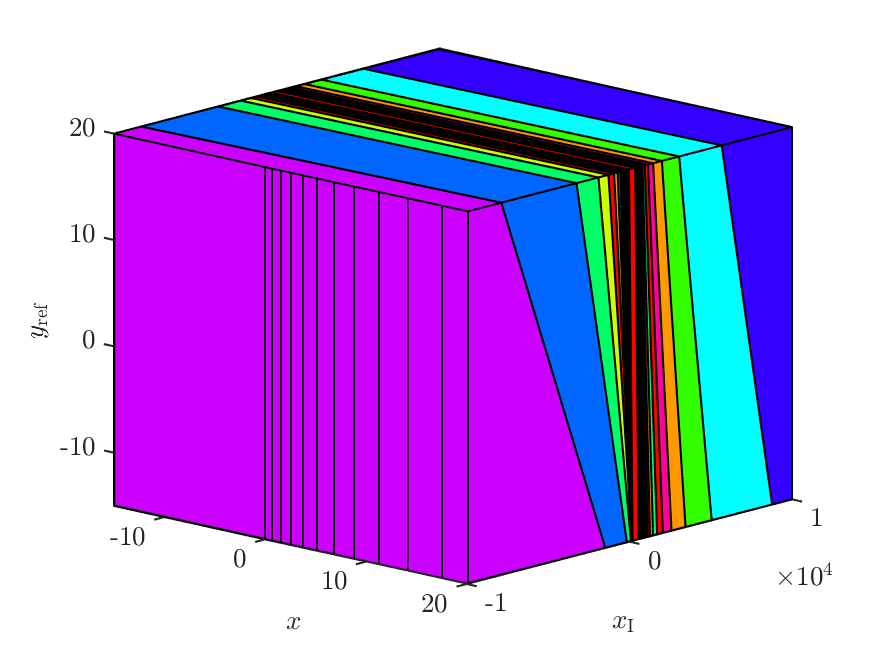}
		\caption{Polytopic partition of the lower boundary explicit MPC.}
	\label{fig:partition_L}
\end{center}
\end{figure}


The designed explicit model predictive controllers were implemented to track a time-varying PWC reference. 
For the initial 200 seconds, the reference temperature was the steady-state value. After that, the reference changed its value twice upwards and twice downwards. The reference changes also acquired different sizes in order to examine the proposed tuning method as it is dependent on the size of the reference step change. Specifically, the reference temperature values were $T_{\mathrm{ref}}$ = 35\,$^{\circ}\mathrm{C}$, 45\,$^{\circ}\mathrm{C}$, 50\,$^{\circ}\mathrm{C}$, 45\,$^{\circ}\mathrm{C}$, 35\,$^{\circ}\mathrm{C}$.

Besides the control design of two boundary explicit MPCs, it was necessary to keep the temperature of the heating medium constant. The heating medium was transported back to the heater after leaving the heat exchanger, i.e., the volume of the heating medium was recycled during the whole operation. The temperature of the heating medium was maintained on the value 70\,$^{\circ}\mathrm{C}$ with a simple proportional controller with proportional gain $P = 20$. The control input from the proportional controller was the electric power which could acquire the values in the range [$0-2$]\,kW and was also normalized to percentage.

The control profiles generated for both considered boundary control setups are compared in Figure~\ref{fig:CV} for the controlled variable, and in Figure~\ref{fig:MV} for the control inputs. 
Note, the constructed explicit MPC controller computed control inputs to respect the constraints on the control inputs and they need not be truncated afterward. 

An interesting phenomenon can be observed while tracking the third reference value, i.e., $T_{\mathrm{ref}}$ = 50\,$^{\circ}\mathrm{C}$. Although the steady-state values of temperature have the same value, the values of the manipulated variable are different. To check the correctness of the results, the measurements were performed multiple times and led to the same behavior. Also, the inlet temperatures of the cold and heating medium were checked to exclude the effect of a disturbance. Regarding the temperature of the cold medium, due to the limited hardware interface, it was not possible to measure the data continuously, store them, and plot the trajectory in a Figure. Nevertheless, the temperature of the cold medium was manually checked multiple times during the experiment and was constant.
	
Regarding the temperature of the heating medium, the corresponding trajectories of the temperature can be seen in Figure~\ref{fig:Th}, and the electric power, i.e., the corresponding manipulated variable, can be seen in Figure~\ref{fig:W}. Note that the legends correspond to the specific setup of MPC, but the temperature of the heating medium was controlled with a simple P controller with the same proportional gain in every control scenario. 

It can be seen that the temperature of the heating medium remains relatively constant during the whole control, except for the undershoots in the scenario with upper boundary MPC, i.e., blue trajectory. The undershoots can be easily associated with the trajectory of the voltage on the pump dosing the heating medium (and ultimately the heating medium flow rate). As the upper boundary MPC calculated ``aggressive'' control inputs, the increased flow rate of the heating medium led to a slight decline in the heating medium temperature. After approximately 100 seconds, the heating medium warmed up to the reference value, i.e., $T_{\mathrm{H, ref}} = 70^{\circ}$\,C and remained constant within the accuracy of the temperature sensor. It can be seen that although the temperature of the heating medium is constant and identical for all control scenarios (MPC setups), the value of the voltage on the pump dosing the heating medium is not the same when tracking the temperature $T_{\mathrm{ref}} = 50^{\circ}$\,C. Therefore, the same conditions were fulfilled for all control scenarios. 

The reason for this behavior could be explained by the peak of the manipulated variable associated with the upper boundary controller at time 800\,s, see Figure~\ref{fig:MV}, blue. After approximately 100\,s, the value of the manipulated variable dropped and settled at a value lower than the value associated with the lower boundary controller, see Figure~\ref{fig:MV}, red. This is a consequence of the heat accumulated inside the heat exchanger plates, and therefore, less heating medium was necessary to heat the cold medium. This phenomenon does not happen when tracking the reference value $T_{\mathrm{ref}}$ = 45\,$^{\circ}\mathrm{C}$, which originates in the nonlinear nature of the heat transfer process. When working in a higher temperature range, the gain of the heat transfer process decreases, and the sensitivity to changes in the heating medium flow is lower. Therefore, even different flow rates of the heating medium lead to the same temperature at the outlet.

\begin{figure}
	\begin{center}
		\includegraphics[width=\textwidth]{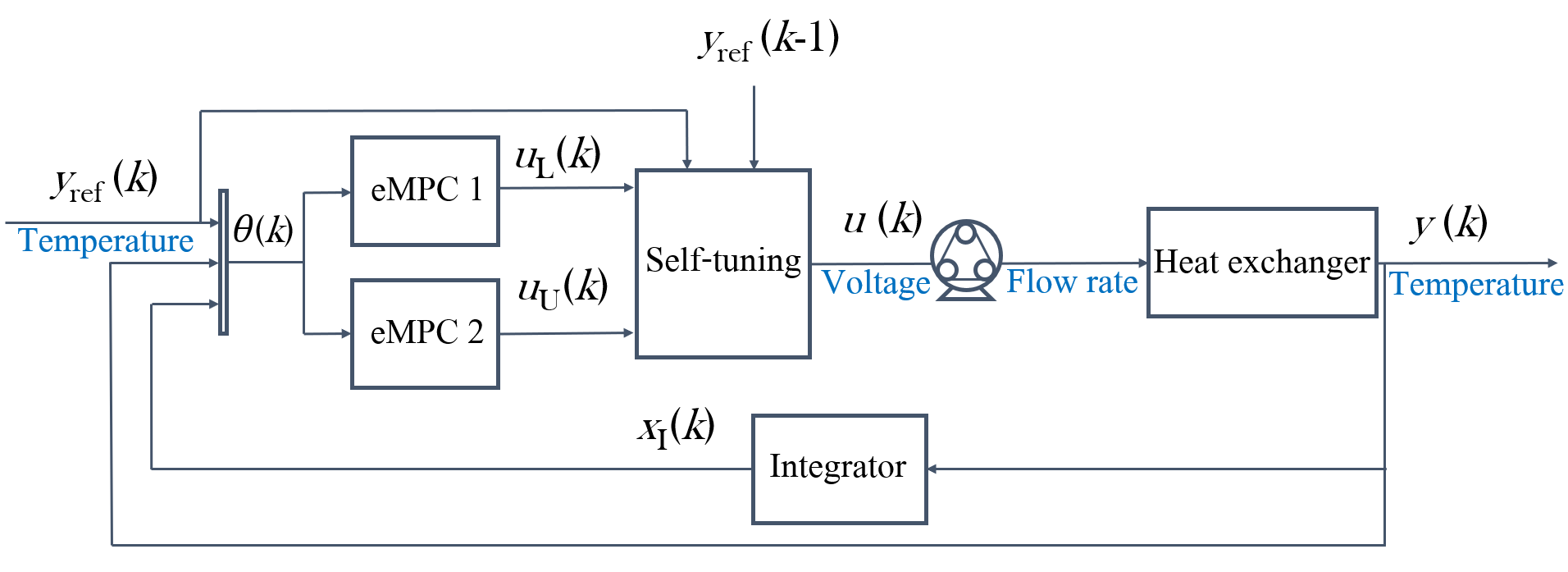}
		\caption[Scheme of the implemented closed-loop control setup.]{Scheme of the implemented closed-loop control setup, where ``eMPC'' denotes explicit MPC.}
		\label{fig:CL}
	\end{center}
\end{figure}



The trajectories in Figure~\ref{fig:CV} show the asymmetric nature of controlling the plant of plate heat exchanger mainly when observing the overshoots and undershoots. When applying the control inputs associated with the lower boundary penalty matrix $Q_\mathrm{y, L}$ in Eq.~\eqref{eq:mpc_problem_cost}, significant undershoots are present when tracking the reference downwards, i.e., when the reference change is negative. On the contrary, when implementing the controller associated with $Q_\mathrm{y, U}$ in Eq.~\eqref{eq:mpc_problem_cost}, the undershoots are negligible, but significant overshoots can be seen when tracking the reference upwards, see Figure~\ref{fig:CV}, blue.

These main experimental observations established the base for the strategy of controller self-tuning. The strategy follows the ideas summarized in Section~\ref{sec:methodology}. Utilizing the nature of the boundary controller with the penalty matrix $Q_\mathrm{y, L}$ is preferred when the reference changes upwards. Therefore, in these operating conditions, the tuning factor is scaled in the first part of the whole interval, i.e., closer to the lower bound. On the contrary, tuning the controller closer to the upper boundary controller associated with $Q_\mathrm{y, U}$ is preferred for negative reference step changes. Therefore, in these operating conditions, the tuning factor is scaled above the splitting value $\rho_{\mathrm{s}}$, i.e., closer to the upper bound. The splitting value of the tuning parameter was chosen simply in the middle of the interval, i.e., $\rho_{\mathrm{s}} = 0.5$. The remaining parameter that needed to be set was the maximal admissible size of the reference step change $\Delta_{\max}$, which was determined to $15\,^{\circ}\mathrm{C}$ as the investigated range of controlled temperature was [$35-50$]\,$^{\circ}\mathrm{C}$. Based on the aforementioned parameters and real-time information about the current reference change, the tuning factor was updated during control. The evolution of the scaled tuning factor $\widetilde{\rho}$ can be seen in Figure~\ref{fig:rho}. When the positive reference changes are tracked, the tuning factor is scaled below the splitting value $\rho_{\mathrm{s}}$. On the contrary, when the reference changes are negative, the tuning factor is scaled above the splitting value $\rho_{\mathrm{s}}$.

The setup of the tuning factor can be associated with tuning of the penalty matrix $Q_{\mathrm{y}}$ according to Eq.~\eqref{eq:tunable_Qy}. The evolution of the penalty matrix $Q_{\mathrm{y}}$ during control is depicted in Figure~\ref{fig:Q}. Note, the penalty matrix evolution in Figure~\ref{fig:Q} does not correspond to tuning of the optimal MPC, but serves for a deeper insight into the association of the interpolated control inputs with the optimal explicit MPC setup.   

\begin{figure}
	\begin{center}
		\includegraphics[width=0.8\textwidth]{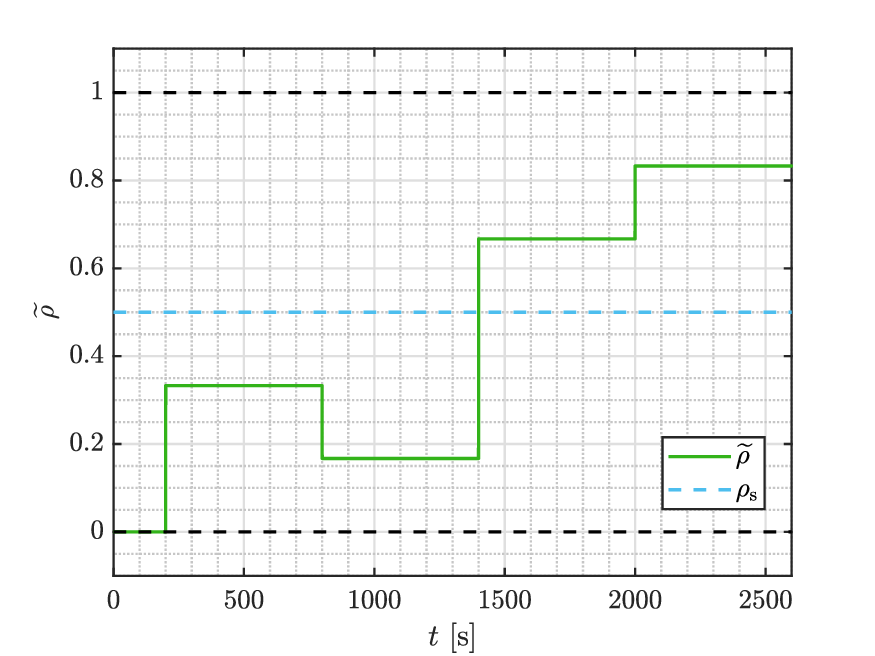}
		\caption{Evolution of the scaled tuning factor $\widetilde{\rho}$ during real-time control. When tracking positive reference changes, the tuning factor is scaled below the splitting value $\rho_{\mathrm{s}}$ (200 -- 1\,400\,s). On the contrary, when the reference changes are negative, the tuning factor is scaled above the splitting value $\rho_{\mathrm{s}}$ (1\,400 -- 2\,600\,s).}
		\label{fig:rho}
	\end{center}
\end{figure}

\begin{figure}
	\begin{center}
		\includegraphics[width=0.8\textwidth]{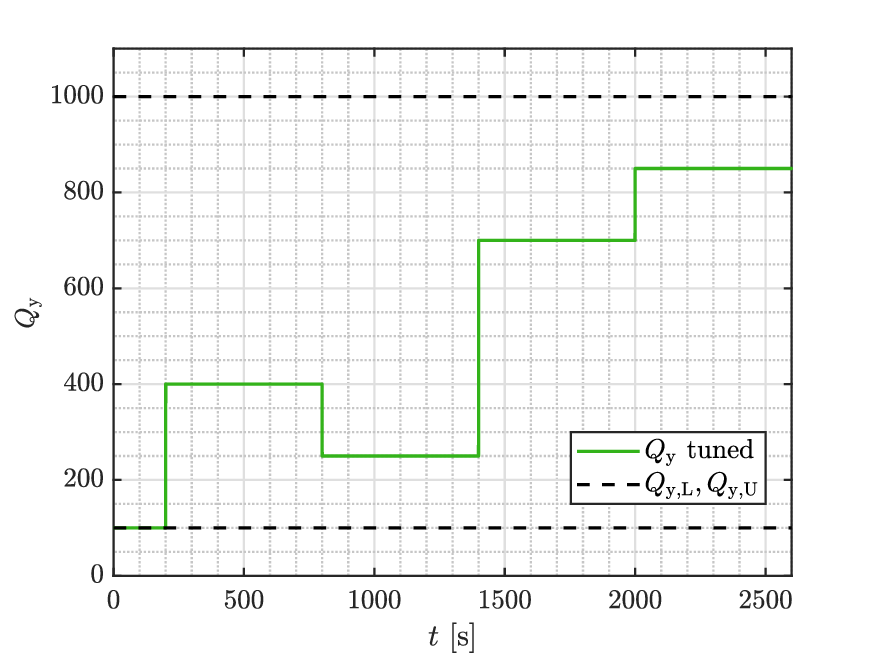}
		\caption{Evolution of the penalty matrix $Q_{\mathrm{y}}$ during real-time control. When tracking positive reference changes, the controller is tuned to operate closer to the lower boundary matrix $Q_\mathrm{y, L}$ (200 -- 1\,400\,s). On the contrary, when the reference changes are negative, the controller is tuned to operate closer to the lower boundary matrix $Q_\mathrm{y, U}$ (1\,400 -- 2\,600\,s).}
		\label{fig:Q}
	\end{center}
\end{figure}

The control input is applied to the system each second, so there is a possible concern regarding the speed at which two explicit MPCs are evaluated. By analyzing the computational speed, it was concluded that the approximate control input can be generated in an average time of 0.01 seconds, which is 100 times faster than the sampling time.

The control results of the self-tunable technique compared to the boundary controllers can be seen in Figure~\ref{fig:CV} for the controlled variable, and in Figure~\ref{fig:MV} for the manipulated variable. It can be seen that the tuned controller combined the benefits of the two boundary controllers. The overshoots and undershoots were reduced, as in the first half of control the penalty matrix $Q_\mathrm{y}$ acquired value from the first half of the penalty interval. When tracking the reference with negative step change, the penalty matrix acquired the values from the second half of the interval, i.e., closer to the upper bound $Q_\mathrm{y, U}$. 
The similarity with the boundary controllers can be seen also on the manipulated variable profiles. Note, the constraints on the input variable were satisfied as they were scaled using linear interpolation based on the boundary controllers which are constructed considering the input constraints. 

\begin{figure}
	\begin{center}
		\includegraphics[width=0.8\textwidth]{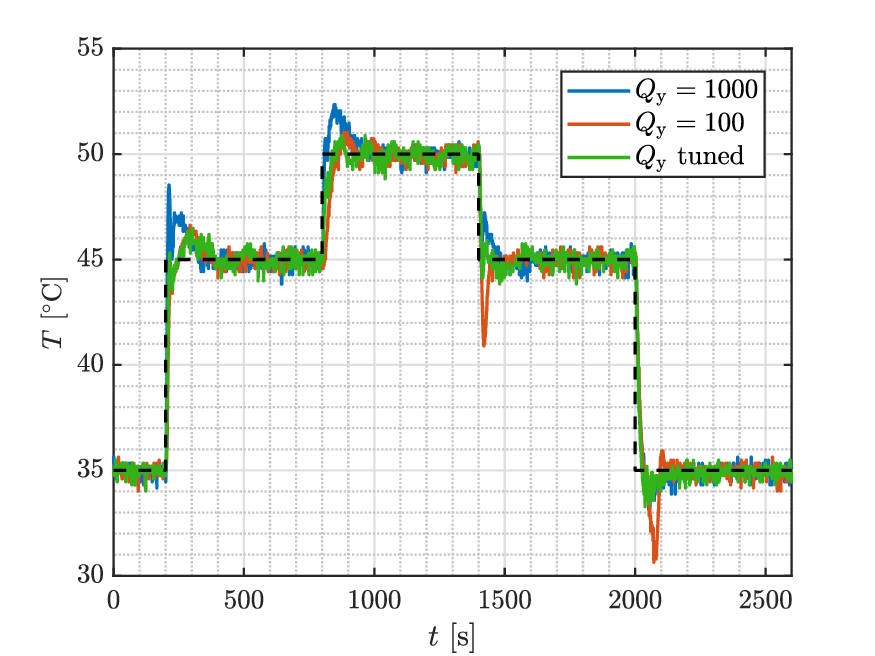}
		\caption{Explicit MPC: Controlled variable trajectory for two boundary controllers and the tuned one. The solid lines represent the controlled temperature $T$ and the dashed line represents the reference value.}
		\label{fig:CV}
	\end{center}
\end{figure}

\begin{figure}
	\begin{center}
		\includegraphics[width=0.8\textwidth]{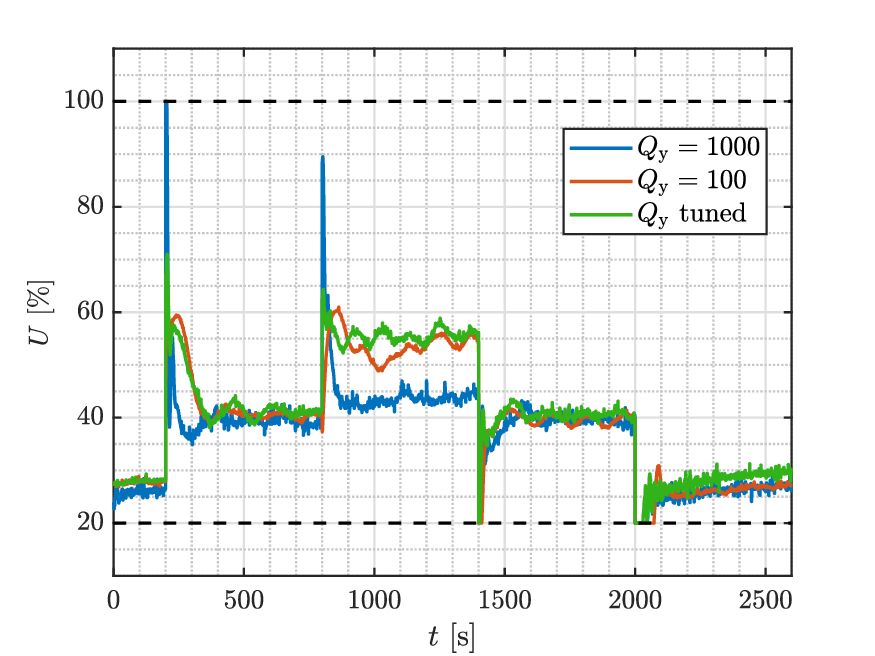}
		\caption{Explicit MPC: Manipulated variable trajectory for two boundary controllers and the tuned one. The solid lines represent the voltage $U$ and the dashed lines represent the constraints.}
		\label{fig:MV}
	\end{center}
\end{figure}

\begin{figure}
	\begin{center}
		\includegraphics[width=0.7\textwidth]{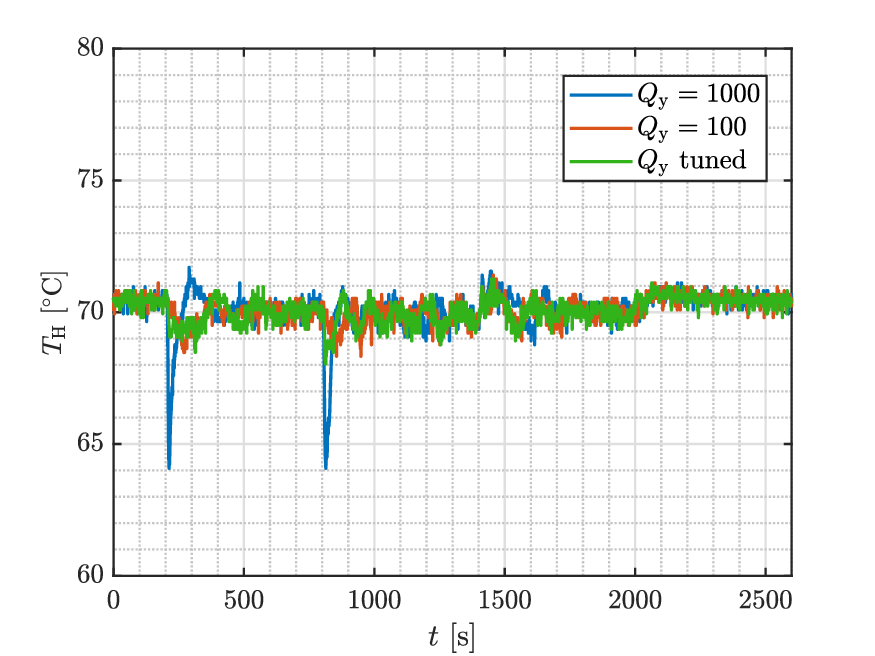}
		\caption{Auxiliary P controller: The trajectory of heating medium temperature control.}
		\label{fig:Th}
	\end{center}
\end{figure}

\begin{figure}
	\begin{center}
		\includegraphics[width=0.7\textwidth]{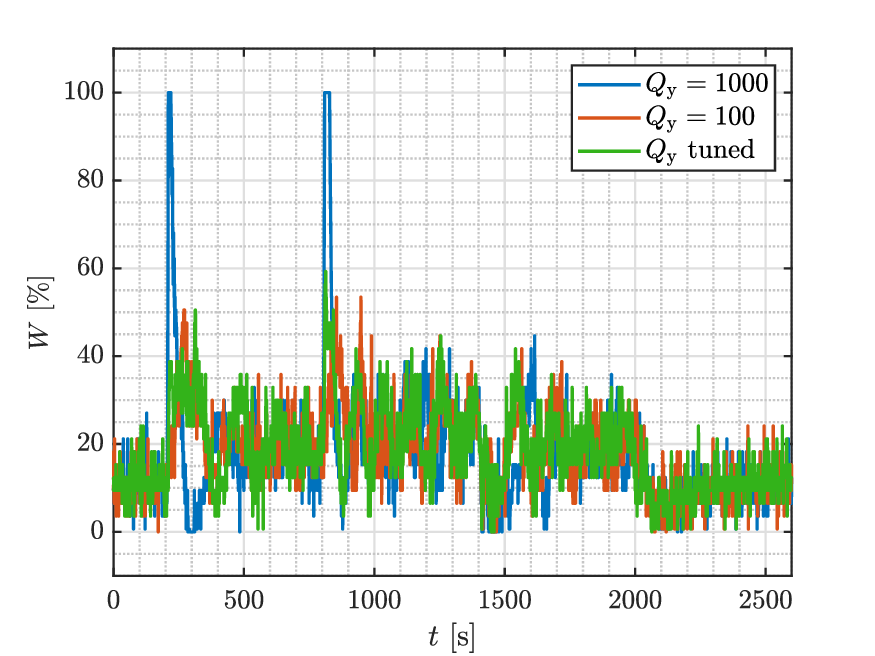}
		\caption{Auxiliary P controller: The trajectory of electric power controlling the heating medium temperature.}
		\label{fig:W}
	\end{center}
\end{figure}

The control performance was also investigated quantitatively. Table~\ref{tab:control_performance} summarizes the evaluated control performance criteria computed for the two boundary controllers and the self-tuned controller. The control performance is evaluated for each reference step change separately. The considered quality criteria are: sum-of-squared control error SSE, maximal overshoot/undershoot $\sigma_{\mathrm{max}}$ and the settling time $t_{\epsilon}$ for 5\,\%-neighbourhood of the reference temperature $T_{\mathrm{ref}}$. To provide better readability of the computed results in Table~\ref{tab:control_performance}, the best values, i.e., the minimum values, are emphasized using a bold font style. 


\begin{table}[h!]
	\begin{center}
		\caption{Control performance criteria.}
		\label{tab:control_performance}
		\begin{tabular}{c|c|c|c|c} 
			Reference step change & $Q_\mathrm{y}$ & SSE [$^{\circ}\mathrm{C}^2$\,s] & $\sigma_{\mathrm{max}}$\,[\%] & $t_{\epsilon}$\,[s]  \\
			\hline
			\multirow{2}{*}{ 35\,$^{\circ}$C $\rightarrow$ 45\,$^{\circ}$C } & 1\,000 & 714 & 33.5 & 16.5 \\
			& 100 & 867 & 16.7 & 12.5 \\ 
			& self-tuned & \textbf{678} & \textbf{15.2} & \textbf{9.5}  \\ 
			\hline
			\multirow{2}{*}{ 45\,$^{\circ}$C $\rightarrow$ 50\,$^{\circ}$C } & 1\,000 & 365 & 47.2 & \textbf{5} \\
			& 100 & 606 & 23.3 & 26.5  \\ 
			& self-tuned & \textbf{248} & \textbf{19.1} & 9.5  \\ 
			\hline
			\multirow{2}{*}{ 50\,$^{\circ}$C $\rightarrow$ 45\,$^{\circ}$C } & 1\,000 & 245 & \textbf{18.9} & \textbf{6.5}  \\
			& 100 & 398 & 79.6 & 31  \\ 
			& self-tuned & \textbf{186} & 24.6 & \textbf{6.5}  \\ 
			\hline
			\multirow{2}{*}{ 45\,$^{\circ}$C $\rightarrow$ 35\,$^{\circ}$C } & 1\,000 & 1\,024 & 18.4 & 22.5  \\
			& 100 & 1\,402 & 41.9 & 90  \\ 
			& self-tuned & \textbf{967} & \textbf{16.5} & \textbf{18.5}   
		\end{tabular}
	\end{center}
\end{table}

As can be seen in Table~\ref{tab:control_performance}, the real-time self-tuning of the explicit MPC controller helped to improve two to three criteria when tracking each reference value. 	
The relative improvement in the percentage, denoted by $\delta$, of using the self-tunable controller is summarized in Table~\ref{tab:improvement} for each reference step change separately. The values were computed as the difference between two criteria values corresponding to the optimal and self-tunable MPC, referred to the self-tunable MPC. The negative numbers represent deterioration of the specific performance criterion in the corresponding reference tracking. 

\begin{table}[h!]
	\begin{center}
		\caption{Relative improvement of the control performance using the self-tunable explicit MPC controller.}
		\label{tab:improvement}
		\begin{tabular}{c|c|c|c|c} 
			& Comparison with $Q_\mathrm{y}$ setup & $\delta$ SSE\,[\%] & $\delta \sigma_{\mathrm{max}}$\,[\%] & $\delta t_{\epsilon}$\,[\%]  \\
			\hline
			\multirow{3}{*}{ 35\,$^{\circ}$C $\rightarrow$ 45\,$^{\circ}$C } & 1000 & 5 &  121 & 74 \\
			& 100 & 28 &  10 & 32 \\ 
			\hline
			\multirow{3}{*}{45\,$^{\circ}$C $\rightarrow$ 50\,$^{\circ}$C } & 1000 & 47 & 147 &$-47$  \\ 
			& 100 & 144 &  22 & 179 \\
			\hline
			\multirow{3}{*}{50\,$^{\circ}$C $\rightarrow$ 45\,$^{\circ}$C } & 1000 & 32 &$-23$& 0 \\ 
			& 100 & 114 & 224 & 377 \\
			\hline
			\multirow{3}{*}{45\,$^{\circ}$C $\rightarrow$ 35\,$^{\circ}$C } &  1000 & 6 & 12 & 22 \\
			& 100 & 45 &  154 & 386 \\
			\hline
			\multirow{3}{*}{Average} &  1000 & 23 & 64 & 12 \\
			& 100 & 83 &  102 & 244
		\end{tabular}
	\end{center}
\end{table}	

Compared to the considered non-self-tunable controllers, the control trajectories and the evaluated quality criteria confirmed the improved control performance for the reference tracking control problem of the heat exchanger with the non-linear and asymmetric behavior. Implementing a self-tunable explicit MPC controller leads to improved control performance in the most analyzed quality criteria, see Table~\ref{tab:improvement}.	 
In average, the control performance criteria improved compared to the upper and lower boundary MPC respectively as follows: the squared-error-based criterion (SSE) reduced by 23\% and 83\%, the maximal overshoot/undershoot $\sigma_{\mathrm{max}}$ reduced by 64\% and 102\%, and the settling time $t_{\epsilon}$ reduced by 12\% and 244\%. 




	In general, utilizing the proposed controller with a scalable aggressiveness according to the operating conditions leads to higher accuracy (lower SSE), lower value of the overshoots (reduced $\sigma_{\mathrm{max}}$), and faster achieving the reference value (decreased $t_{\epsilon}$).

	Obviously, if there exists a well-tuned ``universal'' controller that satisfies the requirements on the control performance in the whole range of the considered operating conditions, then the implementation of the self-tuning procedure is out of scope for such control application. Nevertheless, in numerous practical situations, using only one controller with a constant setup leads to poor or just ``satisfactory'' control results, i.e., the reference value is achieved, but with worse control performance, e.g., leading to high overshoots or settling times. When working on our laboratory case study, a set of different setups of penalty matrices was investigated. In every control scenario, the setup was beneficial only in some working conditions (tracking the reference upwards or downwards). Therefore, the closed-loop control performance is improved by introducing the benefits of the self-tuning method based on the two boundary MPC controllers.

	Note that this strategy relies on a proper design of the two boundary controllers.
	In case a non-negligible disturbance occurs, both boundary controllers
	should be able to solve a disturbance rejection problem as the final
	value of the manipulated variables is interpolated between them. To address the impact of the disturbances directly in constructing the MPC controller design, a robust MPC strategy should be considered, e.g., see~\cite{PR11}. Any robust MPC design method leads to conservative control actions as some portion of the performance is sacrificed to compensate for the impact of the disturbances. Nevertheless, if it is possible to obtain the explicit (multi-parametric) solution of the robust explicit MPC offline, then the same self-tuning procedure is applicable to interpolate between the control actions from the robust controllers. 
	
	\section{Conclusions}
	\label{sec:conclusion}
	
	This paper deals with the experimental implementation and analysis of the novel self-tunable approximated explicit model predictive control method and provides a strategy for an effective self-tuning controller design. Based on the current value of the piece-wise constant reference, the tuning parameter is scaled using linear interpolation. 
	The previously published work related to the self-tunable explicit MPC suggested tuning based on the distance of the reference value from the system steady-state value. This paper presents a novel perspective idea of self-tuning based on the size of reference step change. The self-tuning algorithm aims to compensate for the nonlinear behavior of the controlled system. The self-tuning parameter is updated whenever the reference changes. The tuning value is calculated as the ratio between the size of the reference change and the maximal admissible size of the reference change, which is specified before operation. 
	Another novel contribution addresses the challenging control problem of asymmetric system behavior by splitting the interval of the self-tuning parameter into two ranges, while both intervals are assigned to different operating conditions. The proposed method is implemented on a laboratory-scaled heat exchanger with nonlinear and asymmetric behavior. The asymmetry makes the plant a suitable candidate to analyze the benefits of splitting the interval of the tuning parameter. The decision criterium is negativity or positivity of reference change. When the reference changed upwards, the control input was tuned in the first part of the interval and approached the boundary controller associated with the lower bound on the selected penalty matrix. On the contrary, when the reference changed downwards, the control input was tuned to approach the control input from the boundary controller with the upper bound on the penalty matrix. 
	
	To properly investigate the control results, the control performance was also judged quantitatively using a set of quality criteria. The self-tunable control approach outperformed the conventional control strategy handling just a single controller, i.e., non-tunable controller. In average, the control performance criteria improved compared to the upper and lower boundary MPC respectively as follows: the squared-error-based criterion (SSE) reduced by 23\,\% and 83\,\%, the maximal overshoot/undershoot $\sigma_{\mathrm{max}}$ reduced by 64\,\% and 102\,\%, and the settling time $t_{\epsilon}$ reduced by 12\,\% and 244\,\%.
	
	The approach of the self-tunable technique was successfully implemented on a SISO system but can also be extended to multivariable systems by utilizing only a single value of the tuning parameter $\rho$ to interpolate the values of every control input. It is suggested in Eq.~\eqref{eq:rho_max} that the tuning parameter $\rho$ can be calculated as the maximal value of all the tuning parameters computed for every output reference. However, this is straightforward to implement only for decoupled systems. If there are strong interactions between the system states, the self-tunable technique is challenging to design. In such a case, it is necessary to include expert knowledge about the system state interactions, and the resulting value of the tuning parameter $\rho$ could be computed, e.g., as a weighted average of the individual tuning parameters.
	\section*{Acknowledgments}
	
	The authors gratefully acknowledge the contribution of the Scientific Grant Agency of the Slovak Republic under the grants 1/0545/20, 1/0297/22, and the Slovak Research and Development Agency under the project APVV-20-0261. 
	This research is funded by the European Union’s Horizon Europe under grant no. 101079342 (Fostering Opportunities Towards Slovak Excellence in Advanced Control for Smart Industries). The authors also acknowledge Petronela Belková for help with generating the experimental data.   
	
	\section*{Nomenclature}
	
	\subsection*{Symbols}
		\begin{tabular}{ l l }
			$A$ & system state matrix \\
			$\widetilde{A}$ & augmented system state matrix \\
			$B$ & system input matrix \\
			$\widetilde{B}$ & augmented system input matrix \\
			$C$ & system output matrix \\
			$\widetilde{C}$ & augmented system output matrix \\
			$d_{\max}$ & maximal deviation from the steady-state value \\
			$F$ & slope of the affine control law \\
			$g$ & section of the affine control law \\
			$I$ & identity matrix \\
			$k$ & step of the prediction horizon \\
			$K$ & system gain, $^{\circ}\mathrm{C}$\\
			$N$ & prediction horizon \\
			$n_{\mathrm{u}}$ & size of system inputs \\
			$n_{\mathrm{y}}$ & size of system outputs \\
			$n_{\widetilde{\mathrm{x}}}$ & size of augmented system states \\
			$P$ & proportional gain of proportional controller\\
			$Q_{\mathrm{x}}$ & penalty matrix of the built-in integrator \\
			$Q_{\mathrm{x,L}}$ & lower bound on the penalty matrix of the built-in integrator \\
			$Q_{\mathrm{x,U}}$ & upper bound on the penalty matrix of the built-in integrator \\
			$Q_{\mathrm{y}}$ & penalty matrix of the control error \\
			$Q_{\mathrm{y,L}}$ & lower bound on the penalty matrix of the control error \\
			$Q_{\mathrm{y,U}}$ & upper bound on the penalty matrix of the control error \\
			$R$ & penalty matrix of system inputs \\
			$R_{\mathrm{L}}$ & lower bound on the penalty matrix of system inputs \\
			$R_{\mathrm{U}}$ & upper bound on the penalty matrix of system inputs \\
			$\mathcal{R}$ & critical region \\
			$\mathbb{R}$ & Euclidean space of real numbers \\
			$t$ & time, s \\
			$t_{\epsilon}$ & settling time, s \\
			$T$ & temperature, $^{\circ}\mathrm{C}$ \\
			$T_{\mathrm{C}}$ & cold medium temperature, $^{\circ}\mathrm{C}$ \\
			$T_{\mathrm{H}}$ & heating medium temperature, $^{\circ}\mathrm{C}$ \\
			$T_{\mathrm{H, ref}}$ & heating medium reference temperature, $^{\circ}\mathrm{C}$ \\			
		\end{tabular}
	
	
		\begin{tabular}{ l l }
			$T_{\mathrm{ref}}$ & reference temperature, $^{\circ}\mathrm{C}$ \\
			$T_{\mathrm{s}}$ & sampling time, s \\
			$T^{\mathrm{s}}$ & steady state of temperature, $^{\circ}\mathrm{C}$ \\
			$u$ & control inputs \\
			$u_{\mathrm{L}}$ & control inputs associated with the lower boundary controller\\
			$u_{\mathrm{U}}$ & control inputs associated with the upper boundary controller\\
			$U$ & voltage, \% \\
			$U^{\mathrm{s}}$ & steady state of voltage, \% \\
			$\mathcal{U}$ & set of control inputs \\
			$W$ & electric power, \% \\			
			$x$ & system states \\
			$\widetilde{x}$ & augmented system states \\
			$x_{\mathrm{I}}$ & system states corresponding to the built-in integrator \\
			$y$ & system outputs \\
			$y_\mathrm{\max}$ & maximal value of system outputs \\
			$y_\mathrm{\min}$ & minimal value of system outputs \\
			$y_\mathrm{ref}$ & reference value of system outputs \\
			$\mathcal{Y}$ & set of system outputs \\
			\textit{0} & zero matrix
		\end{tabular}
	
	\subsection{Greek letters}
		\begin{tabular}{ l l }
			$\delta$ & relative improvement, \% \\
			$\Delta_\mathrm{max}$ & maximal size of the reference change, $^{\circ}\mathrm{C}$  \\
			$\Delta_\mathrm{ref}$ & size of the reference change, $^{\circ}\mathrm{C}$ \\
			$\rho$ & tuning factor \\
			$\widetilde{\rho}$ & scaled tuning factor \\
			$\rho_{\mathrm{s}}$ & splitting value of the tuning factor \\
			$\sigma_{\max}$ & maximal overshoot, \% \\
			$\tau$ & system time constant, s\\
			$\theta$ & parameter of optimization problem \\
			$\Theta$ & set of parameter values
		\end{tabular}
	
	\subsection*{Abbreviations}
		\begin{tabular}{ l l }
			eMPC & explicit model predictive control \\			
			LTI  & linear time-invariant (system) \\
			LQR & linear-quadratic regulator \\
			MPC  & model predictive control \\
			mp-QP& multi-parametric quadratic programming (problem) \\
			PID  & proportional–integral–derivative (controller) \\
			PWA  & piece-wise affine (function) \\
			PWC  & piece-wise constant (function) \\
			QP   & quadratic programming (problem) \\
			SISO & single-input and single-output (system) \\
			SSE  & sum-of-squared error 
		\end{tabular}
	
	
	
	
	\bibliographystyle{elsarticle-num} 
	\bibliography{references}

\begin{thebibliography}{10}
\expandafter\ifx\csname url\endcsname\relax
  \def\url#1{\texttt{#1}}\fi
\expandafter\ifx\csname urlprefix\endcsname\relax\def\urlprefix{URL }\fi
\expandafter\ifx\csname href\endcsname\relax
  \def\href#1#2{#2} \def\path#1{#1}\fi

\bibitem{MN20}
M.~M. Morato, J.~E. Normey-Rico, O.~Sename, Model predictive control design for
  linear parameter varying systems: A survey, Annual Reviews in Control 49
  (2020) 64--80.
\newblock \href
  {https://doi.org/https://doi.org/10.1016/j.arcontrol.2020.04.016}
  {\path{doi:https://doi.org/10.1016/j.arcontrol.2020.04.016}}.

\bibitem{YV16}
L.~Zhi-Yong, P.~S. Varbanov, J.~J. Kleme\v{s}, J.~Y. Yong, Recent developments
  in applied thermal engineering: Process integration, heat exchangers,
  enhanced heat transfer, solar thermal energy, combustion and high temperature
  processes and thermal process modelling, Applied Thermal Engineering 105
  (2016) 755--762.
\newblock \href {https://doi.org/10.1016/j.applthermaleng.2016.06.183}
  {\path{doi:10.1016/j.applthermaleng.2016.06.183}}.

\bibitem{KV18}
J.~Kleme\v{s}, P.~Varbanov, Heat transfer improvement, energy saving,
  management and pollution reduction, Energy 162 (2018) 267--271.
\newblock \href {https://doi.org/10.1016/j.energy.2018.08.014}
  {\path{doi:10.1016/j.energy.2018.08.014}}.

\bibitem{RL20}
W.~Roetzel, X.~Luo, D.~Chen, Chapter 1 - heat exchangers and their networks: A
  state-of-the-art survey, in: W.~Roetzel, X.~Luo, D.~Chen (Eds.), Design and
  Operation of Heat Exchangers and their Networks, Academic Press, 2020, pp.
  1--12.
\newblock \href
  {https://doi.org/https://doi.org/10.1016/B978-0-12-817894-2.00001-7}
  {\path{doi:https://doi.org/10.1016/B978-0-12-817894-2.00001-7}}.

\bibitem{AGUEL_fouling}
S.~Aguel, Z.~Meddeb, M.~R. Jeday, Parametric study and modeling of cross-flow
  heat exchanger fouling in phosphoric acid concentration plant using
  artificial neural network, Journal of Process Control 84 (2019) 133--145.
\newblock \href
  {https://doi.org/https://doi.org/10.1016/j.jprocont.2019.10.001}
  {\path{doi:https://doi.org/10.1016/j.jprocont.2019.10.001}}.

\bibitem{TRAFCZYNSKI_fouling}
M.~Trafczynski, M.~Markowski, S.~Alabrudzinski, K.~Urbaniec, The influence of
  fouling on the dynamic behavior of pid-controlled heat exchangers, Applied
  Thermal Engineering 109 (2016) 727--738.
\newblock \href
  {https://doi.org/https://doi.org/10.1016/j.applthermaleng.2016.08.142}
  {\path{doi:https://doi.org/10.1016/j.applthermaleng.2016.08.142}}.

\bibitem{WY18}
Y.~Wang, S.~You, W.~Zheng, H.~Zhang, X.~Zheng, Q.~Miao, State space model and
  robust control of plate heat exchanger for dynamic performance improvement,
  Applied Thermal Engineering 128 (2018) 1588--1604.
\newblock \href
  {https://doi.org/https://doi.org/10.1016/j.applthermaleng.2017.09.120}
  {\path{doi:https://doi.org/10.1016/j.applthermaleng.2017.09.120}}.

\bibitem{VANNIEKERK_tuning}
J.~{van Niekerk}, J.~{le Roux}, I.~Craig, On-line automatic controller tuning
  of a multivariable grinding mill circuit using bayesian optimisation, Journal
  of Process Control 128 (2023) 103008.
\newblock \href
  {https://doi.org/https://doi.org/10.1016/j.jprocont.2023.103008}
  {\path{doi:https://doi.org/10.1016/j.jprocont.2023.103008}}.

\bibitem{MF08}
J.~Mikleš, M.~Fikar, {Process Modelling, Identification, and Control},
  Springer, 2007.

\bibitem{Liptak}
B.~G. Liptak, Instrument Engineers' Handbook, 4th Edition, Vol. 2: Process
  Control and Optimization, CRC Press, London, 2005.
\newblock \href {https://doi.org/https://doi.org/10.1201/9781315219028}
  {\path{doi:https://doi.org/10.1201/9781315219028}}.

\bibitem{Morari_MPC}
M.~Morari, J.~{H. Lee}, Model predictive control: past, present and future,
  Computers \& Chemical Engineering 23~(4) (1999) 667--682.
\newblock \href {https://doi.org/https://doi.org/10.1016/S0098-1354(98)00301-9}
  {\path{doi:https://doi.org/10.1016/S0098-1354(98)00301-9}}.

\bibitem{LQR}
C.~Hajiyev, H.~Soken, S.~Vural, Linear quadratic regulator controller design,
  in: State Estimation and Control for Low-cost Unmanned Aerial Vehicles,
  Springer, Cham, 2015.
\newblock \href
  {https://doi.org/http://dx.doi.org/10.1007/978-3-319-16417-5_10}
  {\path{doi:http://dx.doi.org/10.1007/978-3-319-16417-5_10}}.

\bibitem{Maciejowski_MPC}
J.~Maciejowski, {Predictive Control with Constraints}, Prentice Hall, London,
  2000.

\bibitem{receding_horizon}
J.~Mattingley, Y.~Wang, S.~Boyd, Receding horizon control, IEEE Control Systems
  Magazine 31~(3) (2011) 52--65.
\newblock \href {https://doi.org/https://doi.org/10.1109/MCS.2011.940571}
  {\path{doi:https://doi.org/10.1109/MCS.2011.940571}}.

\bibitem{Vinaya_HE_MPC}
K.~V. Vinaya, K.~Ramkumar, V.~Alagesan, Control of heat exchangers using model
  predictive controller, in: {IEEE}-International Conference On Advances In
  Engineering, Science And Management ({ICAESM} -2012), 2012, pp. 242--246.

\bibitem{Oravec_HE_ATE}
J.~Oravec, M.~Bakošová, A.~Mészáros, N.~Míková, Experimental
  investigation of alternative robust model predictive control of a heat
  exchanger, Applied Thermal Engineering 105 (2016) 774--782.
\newblock \href
  {https://doi.org/https://doi.org/10.1016/j.applthermaleng.2016.05.046}
  {\path{doi:https://doi.org/10.1016/j.applthermaleng.2016.05.046}}.

\bibitem{Gonzales_HE_MPC}
A.~H. González, D.~Odloak, J.~L. Marchetti, Predictive control applied to
  heat-exchanger networks, Chemical Engineering and Processing: Process
  Intensification 45~(8) (2006) 661--671.
\newblock \href {https://doi.org/https://doi.org/10.1016/j.cep.2006.01.010}
  {\path{doi:https://doi.org/10.1016/j.cep.2006.01.010}}.

\bibitem{WC19}
X.~Wu, J.~Chen, L.~Xie, Fast economic nonlinear model predictive control
  strategy of organic rankine cycle for waste heat recovery: Simulation-based
  studies, Energy 180 (2019) 520--534.
\newblock \href {https://doi.org/10.1016/j.energy.2019.05.023}
  {\path{doi:10.1016/j.energy.2019.05.023}}.

\bibitem{DZ18}
Z.~Dong, Z.~Zhang, Y.~Dong, X.~Huang, Multi-layer perception based model
  predictive control for the thermal power of nuclear superheated-steam supply
  systems, Energy 151 (2018) 116--125.
\newblock \href {https://doi.org/10.1016/j.energy.2018.03.046}
  {\path{doi:10.1016/j.energy.2018.03.046}}.

\bibitem{Bemporad_automatica}
A.~Bemporad, M.~Morari, V.~Dua, E.~N. Pistikopoulos, The explicit linear
  quadratic regulator for constrained systems, Automatica 38 (2002) 3--20.
\newblock \href {https://doi.org/https://doi.org/10.1016/S0005-1098(01)00174-1}
  {\path{doi:https://doi.org/10.1016/S0005-1098(01)00174-1}}.

\bibitem{Baric_tunable}
M.~Baric, M.~Baotic, M.~Morari, On-line tuning of controllers for systems with
  constraints, in: Proceedings of the 44th IEEE Conference on Decision and
  Control, 2005, pp. 8288--8293.
\newblock \href {https://doi.org/https://doi.org/10.1109/CDC.2005.1583504}
  {\path{doi:https://doi.org/10.1109/CDC.2005.1583504}}.

\bibitem{Klauco_tunable}
M.~Klau\v{c}o, M.~Kvasnica, Towards on-line tunable explicit {MPC} using
  interpolation, in: Preprints of the 6th IFAC Conference on Nonlinear Model
  Predictive Controle, IFAC, Madison, Wisconsin, USA, 2018.

\bibitem{Oravec_tunable}
J.~Oravec, M.~Klau\v{c}o, Real-time tunable approximated explicit {MPC},
  Automatica 142 (2022) 110315.
\newblock \href
  {https://doi.org/https://doi.org/10.1016/j.automatica.2022.110315}
  {\path{doi:https://doi.org/10.1016/j.automatica.2022.110315}}.

\bibitem{Mayne_stability}
D.~Mayne, J.~Rawlings, C.~Rao, P.~Scokaert, Constrained model predictive
  control: Stability and optimality, Automatica 36~(6) (2000) 789--814.
\newblock \href {https://doi.org/https://doi.org/10.1016/S0005-1098(99)00214-9}
  {\path{doi:https://doi.org/10.1016/S0005-1098(99)00214-9}}.

\bibitem{Kis_NN_MPC}
K.~Kiš, M.~Klaučo, A.~Mészáros, Neural network controllers in chemical
  technologies, in: 2020 IEEE 15th International Conference of System of
  Systems Engineering (SoSE), 2020, pp. 397--402.
\newblock \href
  {https://doi.org/https://doi.org/10.1109/SoSE50414.2020.9130425}
  {\path{doi:https://doi.org/10.1109/SoSE50414.2020.9130425}}.

\bibitem{Kis_NN_MPC_corrector}
K.~Ki\v{s}, P.~Bakar\'a\v{c}, M.~Klau\v{c}o, Nearly optimal tunable {MPC}
  strategies on embedded platforms, in: 18th IFAC Workshop on Control
  Applications of Optimization, IFAC-PapersOnline, 2022, pp. 326--331.
\newblock \href {https://doi.org/https://doi.org/10.1016/j.ifacol.2022.09.045}
  {\path{doi:https://doi.org/10.1016/j.ifacol.2022.09.045}}.

\bibitem{self_tunable}
L.~Gal\v{c}\'ikov\'a, M.~Horv\'athov\'a, J.~Oravec, M.~Bako\v{s}ov\'a,
  Self-tunable approximated explicit model predictive control of a heat
  exchanger, Chemical Engineering Transactions, 2022, Vol. 94~(94) (2022)
  1015--1020.
\newblock \href {https://doi.org/https://doi.org/10.3303/CET2294169}
  {\path{doi:https://doi.org/10.3303/CET2294169}}.

\bibitem{Ruscio_MPC_integral}
D.~D. Ruscio, Model predictive control with integral action: A simple mpc
  algorithm, Modeling, Identification and Control 34~(3) (2013) 119--129.
\newblock \href {https://doi.org/https://doi.org/10.4173/mic.2013.3.2}
  {\path{doi:https://doi.org/10.4173/mic.2013.3.2}}.

\bibitem{Klauco_mpc}
M.~Klau{\v{c}}o, M.~Kvasnica, MPC-Based Reference Governors, Springer, 2019.
\newblock \href {https://doi.org/https://doi.org/10.1007/978-3-030-17405-7}
  {\path{doi:https://doi.org/10.1007/978-3-030-17405-7}}.

\bibitem{pct23}
{Introduction Manual, PCT23-MkII Process Plant Trainer}, Armfield, 2007.

\bibitem{mpt_conf}
M.~Herceg, M.~Kvasnica, C.~Jones, M.~Morari, Multi-parametric toolbox 3.0, in:
  European Control Conference, 2013, pp. 502--510.
\newblock \href {https://doi.org/https://doi.org/10.23919/ECC.2013.6669862}
  {\path{doi:https://doi.org/10.23919/ECC.2013.6669862}}.

\bibitem{PR11}
G.~Pannocchia, J.~B. Rawlings, S.~J. Wright, Conditions under which suboptimal
  nonlinear {MPC} is inherently robust, Systems \& Control Letters 60~(9)
  (2011) 747--755.
\newblock \href
  {https://doi.org/https://doi.org/10.1016/j.sysconle.2011.05.013}
  {\path{doi:https://doi.org/10.1016/j.sysconle.2011.05.013}}.

\end{thebibliography}
	
	
		%
		%
		%
	
\end{document}